\documentclass[journal,onecolumn]{IEEEtran}

\usepackage{amsmath}
\usepackage{amssymb}
\usepackage{amsthm}
\usepackage{amsmath}
\usepackage{amssymb}
\usepackage{url}
\usepackage{amscd}
\usepackage{mathrsfs}
\usepackage[dvips]{graphicx}

\usepackage[OT2,T1]{fontenc}
\DeclareSymbolFont{cyrletters}{OT2}{wncyr}{m}{n}
\DeclareMathSymbol{\Sha}{\mathalpha}{cyrletters}{"58}

\newtheorem{theorem}{Theorem}
\newtheorem{corollary}[theorem]{Corollary}
\newtheorem{lemma}{Lemma}

\newtheorem{defn}{Definition}

\newcommand{\Z}{\mathbb{Z}}
\newcommand{\ZZ}{\Z_{\ge 0}}
\newcommand{\C}{\mathcal{C}}

\newcommand{\cF}{\mathcal{F}}
\newcommand{\ccF}{\F(\!(\!D\!)\!)}
\newcommand{\Fq}{\mathbb{F}_q}

\newcommand{\Fqq}{\mathbb{F}_{q^2}}
\newcommand{\D}{d_{\text{\tiny free}}}
\newcommand{\de}{{\delta}}

\newcommand{\rk}{\mathrm{rank}}
\newcommand{\zv}{{\mathbf{v}}}
\newcommand{\zw}{{\mathbf{w}}}

\newcommand{\grs}{\mathcal{GRS}}
\newcommand{\pp}{\mathcal{P}}

\newcommand{\al}{\alpha}

\newcommand{\ga}{\gamma}
\newcommand{\Ga}{\Gamma}

\newcommand{\ho}{H_{\mathrm{o}}}
\newcommand{\he}{H_{\mathrm{e}}}
\newcommand{\cCo}{\cC_{\mathrm{o}}}
\newcommand{\cCe}{\cC_{\mathrm{e}}}

\newcommand{\wt}{\mathrm{wt}}

\newcommand{\cC}{\mathcal{C}}

\newcommand{\cZ}{\mathcal{Z}}

\newcommand{\ul}{\underline}
\newcommand{\ol}{\overline}

\newcommand{\F}{\mathbb{F}_q}

\newcommand{\wh}{\widetilde{H}}

\begin{document}

\title{Construction of Unit-Memory MDS Convolutional Codes}

\author{Chin Hei Chan and Maosheng Xiong
\thanks{Department of Mathematics, Hong Kong University of Science and Technology, Clear Water Bay, Kowloon, Hong Kong. Email address: Chin Hei Chan (chchanam@connect.ust.hk), Maosheng Xiong (mamsxiong@ust.hk).}
}

\date{}
\maketitle

\begin{abstract}
Maximum-distance separable (MDS) convolutional codes form an optimal family of convolutional codes, the study of which is of great importance. There are very few general algebraic constructions of MDS convolutional codes. In this paper, we construct a large family of unit-memory MDS convolutional codes over $\F$ with flexible parameters. Compared with previous works, the field size $q$ required to define these codes is much smaller. The construction also leads to many new strongly-MDS convolutional codes, an important subclass of MDS convolutional codes proposed and studied in \cite{GL2}. Many examples are presented at the end of the paper.

\end{abstract}


\section{Introduction}

The class of convolutional codes was invented by Elias in 1955 \cite{E} and has been widely in use for wireless, space, and broadcast communications since the 1970s. However, compared with the theory of linear block codes, convolutional codes are not so well understood. In particular, there are only a few algebraic constructions of convolutional codes with good designed parameters.

Let $\F$ be the finite field of order $q$ where $q$ is any prime power. Let $\cF:=\ccF$ be the field of Laurent series over $\Fq$. Following \cite{P2}, a $q$-ary rate $k/n$ and degree $\delta$ convolutional code, or an $(n,k,\delta)_q$ code for short, is a $k$-dimensional subspace of $\cF^n$ over $\cF$ with degree $\delta$. It is known that $\de$ is an invariant and is the external degree of a minimal encode that realizes the convolutional code (see also \cite{JZ,MC}). 

Let $\C$ be an $(n,k,\de)_q$ code and let $\D$ be the free distance of $\C$. The four parameters $n,k,\delta,\D$ are of fundamental importance because $k/n$ is the rate of the code, $\delta$ and $\D$ determine respectively the decoding complexity and the error correcting capability of $\C$ with respect to some decoding algorithms such as the celebrated Viterbi algorithm \cite{VI}. For these reasons, for given rate $k/n$ and $q$, generally speaking, it is desirable to construct convolutional codes with relatively small degree $\delta$ and relatively large free distance $\D$. The generalized Singleton bound for an $(n,k,\delta)_q$ convolutional code $\C$, in its most general form, proposed and proved by Rosenthal and Smarandache \cite{R}, states that the free distance $\D$ of $\C$ must satisfy
\begin{eqnarray} \label{1:mds} \D \le (n-k) \left(\left\lfloor \frac{\delta}{k}\right\rfloor+1\right)+\delta+1.\end{eqnarray}
If the inequality is attained as an equality, then $\C$ is called a maximal-distance-separable (MDS) convolutional code. As in the classical case, MDS convolutional codes form an optimal family of convolutional codes, the study of which is of great importance.

For any rate $k/n$ and any degree $\delta$, Rosenthal and Smarandache \cite{R} established the existence of $(n,k,\delta)_q$ MDS convolutional codes over some finite field $\F$ by techniques from algebraic geometry without giving explicit constructions. Then in a beautiful follow-up paper \cite{S}, building upon ideas from Justesen \cite{JUS3}, the authors provided an explicit construction of MDS convolutional codes for each rate $k/n$ and each degree $\delta$ over some $\F$. However, in their construction, a relatively large field size $q$ is required: it is necessary that $n|(q-1)$ and $n \le \frac{q-1}{2}$. Hence the question was raised by the authors as to whether or not it is possible to come up with new constructions of MDS convolutional codes so that the field size $q$ can be reduced and the condition $n|(q-1)$ can be dropped.

In this paper we provide a partial but affirmative answer to this question. Roughly speaking, we construct unit-memory $(n,k,\de)_q$ MDS convolutional codes for any rate $k/n$ and relative small but flexible $\de$ as long as $n \le q+1$, hence the strength of the result is almost comparable to that of classical MDS linear block codes over $\F$. Moreover, a nice feature of the construction is that many of the codes constructed satisfy a much stronger MDS condition: they are actually ``strongly-MDS'' convolutional codes, an important subclass of MDS convolutional codes proposed and studied by Gluesing-Luerssen, Rosenthal and Smarandache in \cite{GL2}. We summarize our results as follows.

\begin{theorem}\label{1:main} Let $q$ be any prime power.

(i). There exists an $(n,k+\de,\de)_q$ MDS convolutional code for any $q \geq 3$ and positive integers $n,\de, k$ such that $2 \de+k \le n \le q$.

(ii). There exists an $(n,k+\de,\de)_q$ strongly-MDS convolutional code for any $q \geq 3$ and positive integers $n,\de,k$ such that $3 \de+k-1 \le n \le q$.

(iii). If $q>2$ is even, then there exists a $(q+1,q+2-2k-\de,\de)_q$ MDS convolutional code for any positive integers $k,\de$ such that $\de+k \le \frac{q+1}{2}$.

(iv). There exists a $(q+1,k+2 \de,2 \de)_q$ strongly-MDS convolutional code for any $q \geq 5$ and positive integers $k,\de$ such that $6\de+k \le q+2$.
\end{theorem}

We remark that (i) and (ii) of Theorem \ref{1:main} summarize Theorem \ref{3:thm1} in Section \ref{sec3} and Theorem \ref{4:thm2} in Section \ref{sec4}; (iv) summarizes Theorems \ref{5:thm3} and \ref{5:thm4} in Section \ref{sec5}; (iii) is obtained by using Theorem \ref{51:thm4} in Section \ref{sec5} and Lemma \ref{2:lem4} in Section \ref{sec2}. It shall be noted that in (ii) and (iv) the codes also have a maximum distance profile. Interested readers may review these theorems for details.

It was proved in \cite[Theorem 3.11]{GL2} that for every positive integers $n,k,\de$ such that $n-k$ divides $\de$ and for every prime number $p$ there exists a strongly-MDS code with parameters $(n,k,\de)$ over a suitably large field of characteristic $p$, and it was conjectured (see \cite[Conjecture 3.13]{GL2}) that for all $n>k>0$ and for all $\de \ge 0$ there exists an $(n,k,\de)_q$ code over a sufficiently large field which is both strongly-MDS and has a maximum distance profile. (ii) and (iv) of Theorem \ref{1:main} can be considered as a small step further toward this conjecture.

In this paper we only consider construction of unit-memory MDS convolutional codes, usually over a large alphabet, but this class of codes should be of great interest both in theory and in applications. First, unit-memory convolutional codes (UMC) are an interesting class of convolutional codes because their block length can be chosen to agree with the word length of computers or microprocessors that are used in the coding and decoding process. Lee \cite{LEE} suggested that short binary UMC's are attractive with Viterbi decoding as the inner coding component of a concatenated system. Thommesen and Justesen \cite{T} further derived bounds on the distance profile and free distance of binary UMC's and suggested that UMC's may have superior properties. Ebel \cite{EB} has made a thorough discussion on various methods in search of UMC's with good parameters. Second, as was pointed out in \cite{GL2}, a convolutional code over a finite alphabet can be practically identified with a finite linear state machine (LFSM) which has redundancy and which is capable of correcting processing errors. In a series of recent papers \cite{HA}--\cite{HA3}, Hadjicostis and Verghese showed how to error protect a given LFSM with a larger redundant LFSM capable of detecting and correcting state transition errors, and the construction of such redundant system boils down to constructing convolutional codes with good free distance over an alphabet which is in general not binary. With this respect, strongly-MDS convolutional codes constructed in this paper are particularly suited and may have potential for such applications.

Our technique of constructing MDS convolutional codes relies on the idea proposed by Aly, Grassl, Klappenecker, R\"otteler and Sarvepalli \cite[Theorem 3]{ALY}, which generalized the famous method of Piret \cite{P2} and provided a powerful machinery to construct convolutional codes from linear block codes. While this method has been employed multiple times by various authors (\cite{CH},\cite{GU0}-\cite{GU2},\cite{YA},\cite{ZH}) to obtain some new classical and quantum MDS convolutional codes, the MDS convolutional codes obtained in these works usually have degree not larger than 2. Our contribution in this paper is the realization that by carefully adopting this method \cite[Theorem 3]{ALY}, coupled with cyclic codes of general type \cite{LXG}, we can construct unit-memory MDS convolutional codes with fairly large degree. 


Let us now give the structure of the paper. In Section \ref{sec2} we review the basic theories of cyclic codes of general type and convolutional codes, in particular ``strongly-MDS''  convolutional codes, column distances and \cite[Theorem 3]{ALY}, all of which will turn out to be useful in the paper. As for the construction of unit-memory MDS convolutional codes $(n,k,\de)_q$, we treat the case $n \le q-1$ in Section \ref{sec3}, the case $n=q$ in Section \ref{sec4}, and the case $n=q+1$ in Section \ref{sec5}. In Section \ref{sec6} we provide several explicit examples of MDS convolutional codes over the field $\mathbb{F}_8$. In Section \ref{conclude} we conclude the paper.

\section{Preliminaries}\label{sec2}
In this section we review some basic theories of cyclic codes of general type and convolutional codes.

\subsection{Cyclic codes of general type}
Throughout the paper we fix some standard notation.

Let $\F$ be the finite field of order $q$, where $q$ is any prime power. An $[n,k,d]_q$ code is a $k$-dimensional subspace of $\F^n$ over $\F$ with minimum Hamming distance $d$. The well-known Singleton bound of an $[n,k,d]_q$ code is $d \le n-k+1$. If the inequality is achieved, it is called a maximal-distance-separable (MDS) code.

\begin{defn}
Let $0 \ne f(x) \in \F[x]$ be a monic polynomial. A nonzero ideal $\cC$ of the principal ideal ring $\F[x]/(f(x))$ is called a cyclic code of type $f$, or an $f$-cyclic code for short.
\end{defn}

The concept of $f$-cyclic codes naturally generalizes the concept of cyclic, negacyclic and constacyclic codes, which have proved quite useful in many recent works. The basic theory of $f$-cyclic codes is very similar to that of cyclic, negacyclic and constacylic codes as well and hence must have been known for a long time. In particular, for any $f$-cyclic code $\cC$, there is a unique monic polynomial $g \in \F[x]$ of least degree such that $g(x)|f(x)$ and $\cC=(g(x))$. This $g(x)$ is called the generator polynomial of $\cC$. Because of lacking proper references, $f$-cyclic codes was introduced in \cite{LXG} where some of their properties were derived, with the sole purpose of providing more flexibility in constructing MDS block codes. We are grateful to Patrick Sol\'e, who kindly pointed out that $f$-cyclic codes were called pseudo-cyclic codes and polycyclic codes in the literature which have been investigated, though the results which we collect below (see \cite[Theorems 9 and 10]{LXG}) may not be so easy to locate in the references.

\begin{lemma}\label{2:lem1}
Let $f \in \F[x]$ be a monic polynomial, $\deg f=n \ge 1$ and $f(x) \mid \left(x^{q-1}-1\right)$. Let $\theta$ be a primitive element of $\F$. Let $g(x)=\prod_{i=d_1}^{d_2} (x-\al_i)$, where $\al_i=\theta^{a+bi}$, and $d_1,d_2,a,b$ are some fixed integers. Assume that $g(x)|f(x)$. Let $\cC$ be an $f$-cyclic code over $\F$ with generator polynomial $g$. If
\begin{equation}\label{con1}
\frac{q-1}{(b,q-1)} \ge n \ge d_2-d_1+1 \ge 1,
\end{equation}
then $\cC$ is an $[n,n-d_2+d_1-1,d_2-d_1+2]_{q}$ MDS code.
\end{lemma}

\begin{lemma}\label{2:lem2}
Let $f \in \F[x]$ be a monic polynomial, $\deg f=n \ge 1$ and $f(x) \mid \left(x^{q^2-1}-1\right)$. Let $\theta$ be a primitive element of $\Fqq$. Let $g(x)=\prod_{i=d_1}^{d_2} (x-\al_i)$, where $\al_i=\theta^{a+bi}$, and $d_1,d_2,a,b$ are some fixed integers. Assume that $g(x)|f(x)$.
\begin{enumerate}
\item[(i)] $g \in \F[x]$ if and only if for any $i$ ($d_1 \le i \le d_2$), there exists a $j$ ($d_1 \le j \le d_2$) such that
    \begin{eqnarray} \label{con02}
    q(a+bi) \equiv a+bj \pmod{q^2-1}.
    \end{eqnarray}

\item[(ii)] Assume that $g \in \F[x]$. Let $\cC$ be an $f$-cyclic code over $\F$ with generator polynomial $g$. If
\begin{equation}\label{con12}
\frac{q^2-1}{(b,q^2-1)} \ge n \ge d_2-d_1+1 \ge 1,
\end{equation}
then $\cC$ is an $[n,n-d_2+d_1-1,d_2-d_1+2]_{q}$ MDS code.
\end{enumerate}
\end{lemma}

\subsection{Convolutional codes}

In this section we first introduce the concept of ``strongly-MDS'' convolutional codes, which was proposed and studied in the beautiful paper \cite{GL2}. We follow the presentation of \cite{GL2}. For basic theories of convolutional codes, interested readers may consult the excellent textbooks \cite{JZ,MC,P2}.

As in the introduction, $\cF=\F(\!(\!D\!)\!)$ is the field of Laurent series over the finite field $\F$. Let $\cC$ be an $(n,k,\de)_q$ convolutional code of memory $\nu$, then there is a $k \times n$ minimal encoder of the form
\[G=\sum_{j=0}^{\nu}G_jD^j \in \F[D]^{k \times n}, G_j \in \F^{k \times n}, G_{\nu} \ne 0,\]
here $\nu$ is the memory of the code, and an $(n-k) \times n$ minimal parity-check matrix of the form
\[H=\sum_{j=0}^{\mu}H_jD^j \in \F[D]^{(n-k) \times n}, H_j \in \F^{(n-k) \times n}, H_{\mu} \ne 0,\]
such that
\[\cC=\left\{aG:a \in \cF^k\right\}=\left\{v \in \cF^n: H v^T=0\right\}.\]
Both $G$ and $H$ are of full rank and satisfy $GH^T=0$. The weight of $v=\sum_{j=r}^{\infty}v_jD^j \in \cF^n$ where $v_j \in \F^n$ is defined as
\[\wt(v):=\sum_{j=r}^{\infty} \wt(v_j).\]
Here $\wt(v_j)$ denotes the Hamming weight of $v_j \in \F^n$. The free distance of the code $\cC \in \cF^n$ is defined as
\[\D:=\min\left\{\wt(v): v \in \cC, v \ne 0\right\}.\]

For every $j \in \ZZ$, the truncated sliding generator matrices $G_j^c \in \F^{(j+1)k \times (j+1) n}$ and parity-check matrices $H_j^c \in \F^{(j+1)(n-k) \times (j+1)n}$ are given by

\[G_j^c:=\left[\begin{array}{cccc}
G_0&G_1&\cdots&G_j\\
&G_0 & \cdots &G_{j-1}\\
&& \ddots &\vdots\\
&&&G_0\end{array}\right],\]

\[H_j^c:=\left[\begin{array}{cccc}
H_0&&&\\
H_1&H_0 & &\\
\vdots&\vdots& \ddots & \\
H_j&H_{j-1}&\cdots&H_0\end{array}\right],\]
where we let $G_j=0$ (resp. $H_j=0$) whenever $j>\nu$ (resp., $j > \mu$); see also \cite[pp. 110]{JZ}. The $j$-th \emph{column distance} of $\cC$ is given by
\begin{eqnarray*}d_j^c&=&\min\left\{\wt\left((u_0,\ldots,u_j)G_j^c\right): u_i \in \F^k, u_0 \ne 0\right\}\\
&=& \min\left\{\wt\left(\hat{v}\right): \hat{v}=\left(\hat{v}_0,\ldots,\hat{v}_j\right)\in \F^{(j+1)n}: H_j^c \left(\hat{v}\right)^T=0, \hat{v}_0 \ne 0\right\}.
\end{eqnarray*}
The $(\nu+1)$-tuple of the numbers $\left(d_0^{\nu},d_1^c,\ldots,d_{\nu}^c\right)$ is called the \emph{distance profile} of the code \cite[pp. 112]{JZ}. The column distances are invariants of the code $\cC$ (\cite[Sec. 3.1]{JZ}) and satisfy
\begin{itemize}
\item[(1).] $d_0^c \le d_1^c \le d_2^c \le \ldots \le \lim_{j \to \infty} d_j^c=\D$.

\item[(2).] $d_j^c \le (n-k)(j+1)+1$ for every $j \in \ZZ$ (\cite[Proposition 2.2]{S});

\item[(3).] If $d_j^c = (n-k)(j+1)+1$ for some $j \in \ZZ$, then $d_i^c \le (n-k)(i+1)+1$ for all $i \le j$ (\cite[Corollary 2.3]{S});

\item[(4).] The generalized Singleton bound: $\D \le (n-k)\left(\lfloor \de/k \rfloor+1\right)+\de+1$ (\cite[Theorem 2.6]{S});

\item[(5).] Let $M:=\min\left\{j \in \ZZ: d_j^c=\D\right\}$, then $M \ge \lfloor \de/k \rfloor +\lceil \de/(n-k) \rceil$ (\cite[Proposition 2.7]{S}).
\end{itemize}

From Property (5), strongly-MDS convolutional codes are MDS convolutional codes such that the generalized Singleton bound is attained by the earliest column distance possible, that is (\cite[Definition 2.8]{S}):

\begin{defn} An $(n,k,\de)_q$ code with column distance $d_j^c, j \in \ZZ$, is called \emph{strongly-MDS}, if
\begin{eqnarray}  \label{2:smds}
d_M^c=(n-k) \left(\left\lfloor \de/k \right\rfloor+1\right)+\delta+1, \mbox{ for } M=\left\lfloor \de/k\right\rfloor+\left\lceil \de/(n-k)\right\rceil.  \end{eqnarray}
\end{defn}
The concept of ``maximal distance profile'' was defined in \cite{GL2} and is related to the notion of optimum distance profile (ODP), see \cite[pp. 112]{JZ}.

\begin{defn} An $(n,k,\de)_q$ code with column distance $d_j^c, j \in \ZZ$, is said to have a \emph{maximal distance profile}, if
\begin{eqnarray}  \label{2:mdp}
d_L^c=(n-k) \left(L+1\right)+1, \mbox{ where } L=\left\lfloor \de/k\right\rfloor+\left\lfloor \de/(n-k)\right\rfloor.  \end{eqnarray}
\end{defn}

It was proved that an $(n,k,\de)_q$ code $\cC$ has a maximum distance profile if and only if the dual code $\cC^{\bot}$ has this property \cite[Theorem 5.2]{GL2}.

Next we recall \cite[Theorem 3]{ALY} which provides a most powerful machinery to construct convolutional codes from linear block codes. For the sake of clarity and simplicity, we only quote the parts of \cite[Theorem 3]{ALY} which are most relevant to this paper. The statement of the theorem is written in slightly different form. Interested readers shall consult the original paper for details.
\begin{theorem} \label{2:con}
Let $\cC$ be an $[n,k,d]_q$ linear code with parity check matrix $H \in \F^{(n-k) \times n}$ where $1 \le k \le n-2$. Assume that $H$ is partitioned into submatrices $H_0,H_1,\ldots,H_m$ as $H^T=[H_0^T,H_1^T, \cdots,H_m^T]$ such that $\kappa=\rk H_0$ and $\rk H_i \le \kappa$ for $1 \le i \le m$ and $\kappa \ge \rk H_m >0$. Define the polynomial matrix $G(D)$ as
\[G(D)=\wh_0+\wh_1D+\cdots+\wh_m D^m \in \F[D]^{\kappa \times n},\]
where the matrices $\wh_i$ are obtained from $H_i$ by adding some zero-rows so that the $\wh_i$ all have $\kappa$ rows in total. Then
 \begin{itemize}
 \item[(a)] The matrix $G(D)$ is a $\kappa \times n$ minimal encoder.

 \item[(b)] Let $V$ be the convolutional code with $G(D)$ as the parity-check matrix, that is, $V=\left\{v \in \cF^n: G(D) v^T=0\right\}$. Let $\D$ be the free distance of $V$. Then
     \begin{eqnarray*} \min\left\{d_0+d_m,d\right\} \le \D \le d,\end{eqnarray*}
where $d_i$ is the minimum distance of the code $\cC_i:=\left\{v \in \F^n: H_i v^T=0\right\}$.
 \end{itemize}
\end{theorem}

The following lemma which is a direct consequence of Theorem \ref{2:con}, provides guidance as to how to construct unit-memory MDS convolutional codes.
\begin{lemma} \label{2:lem3}
Let $\cC$ be an $[n,k,n-k+1]_q$ MDS linear code with parity check matrix $H \in \F^{(n-k) \times n}$ where $1 \le k \le n-2$. Assume that $H$ is partitioned into two submatrices $H_0,H_1$ as $H=\left[\begin{array}{c}H_0\\
H_1\end{array}\right]$ such that $\ga=\rk H_0 \ge \de =\rk H_1>0$. Define
\[G(D)=\wh_0+\wh_1D,\]
where $\wh_0=H_0$ and $\wh_1$ is obtained from $H_1$ by adding some zero-rows so that $\wh_i$ has $\ga$ rows in total. Then
 \begin{itemize}
 \item[(a)] The matrix $G(D)$ is a $\ga \times n$ minimal encoder.

 \item[(b)] The convolutional code $V$ obtained from $G(D)$ as a parity-check matrix has parameters $(n,k+\de,\de)_q$.

 \item[(c)] Let $\cC_i$ be the linear code with parity-check matrix $H_i, i=0,1$. If both $\cC_i$ are MDS codes, then $V$ is a unit-memory MDS convolutional code.
 \end{itemize}
\end{lemma}
\begin{proof}
Since $G(D)$ is a minimal encoder, the dual code of $V$ has parameters $(n,\ga,\de)_q$, hence $V$ has parameters $(n,n-\ga,\de)_q=(n,k+\de,\de)_q$ (\cite[Theorem 2.66]{JZ}). If both $\cC_i$ are MDS, then $d_0 \ge \ga+1, d_1 \ge \de+1$, hence $d_0+d_1 \ge n-k+1$ and we obtain $\D=n-k+1$ from Theorem \ref{2:con}. By the generalized Singleton bound (\ref{1:mds}), $V$ is an MDS convolutional code. It is easy to see that $V$ is of unit-memory.

\end{proof}

The next lemma, adopted from \cite[Theorem 28]{ROX}, indicates how to obtain new MDS convolutional codes from old ones.

\begin{lemma} \label{2:lem4} Let $\cC$ be an MDS $(n,k,\de)_q$ convolutional code with minimal encoder $G(D) \in \F[D]^{k \times n}$ and row indices $v=v_1=\ldots=v_l<v_{l+1}=\ldots=v_{k}=v+1$. Let $\ol{G}(D) \in \F[D]^{(k-1) \times n}$ be the matrix obtained from $G(D)$ by omitting any of the last $k-l$ rows of $G(D)$, and let $\ol{\cC}$ be the convolutional code generated by the encoder $\ol{G}(D)$. Then $\ol{\cC}$ is also an MDS convolutional code.
\end{lemma}


\section{Unit-memory MDS $(n,k,\de)_q$ codes with $n \le q-1$} \label{sec3}

Let $q \geq 3$ be a prime power. For any given positive integers $n$ and $k$ such that $k <n  \le q-1$, define
\[f(x)=\prod_{i=0}^{n-1} \left(x-\theta^i\right), \quad g(x)=\prod_{i=0}^{n-k-1} \left(x-\theta^i\right),\]
where $\theta$ is a primitive element of $\F$. Let $\cC$ be an $f$-cyclic code with generator polynomial $g(x)$, that is, $\cC=(g(x)) \subset \F[x]/(f(x))$. Then each $v(x)=\sum_{i=0}^{n-1}v_i x^i \in \cC$ if and only if $g(x)|v(x)$, that is, \begin{eqnarray} \label{3:gx}v\left(\theta^j\right)=0, \quad \mbox{ for } \forall \, 0 \le j \le n-k-1.\end{eqnarray}
By Lemma \ref{2:lem1}, $\cC$ is an $[n,k,n-k+1]_q$ MDS linear code. Define the $(n-k) \times n$ matrix $H$ by
\[H^T=\left[\ul{h}_0^T,\ul{h}_1^T, \ldots,\ul{h}_{n-k-1}^T\right],\]
where each $\ul{h}_j$ is a $1 \times n$ row-vector with entries in $\F$ given by
\[\ul{h}_j=\left[1,\theta^{j},\theta^{2j}, \ldots,\theta^{(n-1)j}\right].\]
It is easy to see that the condition (\ref{3:gx}) is equivalent to $H v^T=0$, where $v=(v_0,v_1,\ldots,v_{n-1}) \in \F^n$. Therefore $H$ is a parity-check matrix of $\cC$.

Let $\ga,\de$ be any positive integers such that $\ga \ge \de$ and $\ga +\de=n-k$. Split the matrix $H$ into 2 disjoint submatrices $H_0$ and $H_1$ such that
$H=\begin{bmatrix}
 H_0 \\
 H_1
\end{bmatrix}$ where
\[H_0^T=\left[\ul{h}_0^T,\ul{h}_1^T,\ldots,\ul{h}_{\ga-1}^T\right], \quad H_1^T=\left[\ul{h}_{\ga}^T,\ul{h}_{\ga+1}^T, \ldots,\ul{h}_{\ga+\de-1}^T\right].\]

It shall be noted that any $\ga$ columns of $H_0$ (resp. any $\de$ columns of $H_1$) are linearly independent, so let $\cC_i$ be the linear code over $\F$ by the parity-check matrix $H_i, i=0,1$, then each $\cC_i$ is an MDS code.

Define two $\ga \times n$ matrices $\wh_0,\wh_1$ by
\begin{eqnarray*}
\wh_0=H_0, \quad \wh_1^T=\left[\mathbf{0}_{n \times (\ga-\de)},\ul{h}_{\ga}^T,\ul{h}_{\ga+1}^T, \ldots,\ul{h}_{\ga+\de-1}^T\right].
\end{eqnarray*}
Here $\mathbf{0}_{r \times s}$ denotes the $r \times s$ zero matrix for any positive integers $r$ and $s$. Let $V \in \cF^n$ be the convolutional code defined by the parity-check matrix $G(D)$ given by
\[g(D)=\wh_0+\wh_1D.\]
We prove the following result.

\begin{theorem} \label{3:thm1}
(i). The code $V$ defined above is an $(n,k+\de,\de)_q$ unit-memory MDS convolutional code whenever
\[1 \le k <n \le q-1, \quad 1 \le \de \le \frac{n-k}{2}\,.\]

(ii). The code $V$ is an $(n,k+\de,\de)_q$ unit-memory MDS convolutional code with a maximal distance profile whenever

\[1 \le k <n \le q-1, \quad 1 \le \de < \frac{n-k}{2}\,.\]

(iii). The code $V$ defined above is an $(n,k+\de,\de)_q$ unit-memory strongly-MDS convolutional code whenever
\[1 \le k <n \le q-1, \quad 1 \le \de \le \frac{n-k+1}{3}\,.\]
\end{theorem}

\begin{proof}
(i) follows directly from Lemma \ref{2:lem3}, as $\cC_0,\cC_1$ are both MDS linear codes.

(ii) is obvious since $L=\lfloor \de/(k+\de)\rfloor+\lfloor \de/(n-k-\de)\rfloor=0$.

As for (iii), since $M=\lfloor \de/(k+\de)\rfloor+\lceil \de/(n-k-\de)\rceil=1$ and $d_1^c \le \D=d=n-k+1$, it suffices to prove that
\begin{eqnarray}\label{2:dineq} d_1^c \ge d=n-k+1,\end{eqnarray}
where $d_1^c$ is the $1$-st column distance of $\cC$ which is given by
\begin{eqnarray*}d_1^c
&=& \min\left\{\wt\left(\hat{v}\right): \hat{v}=\left(\hat{v}_0,\hat{v}_1\right)\in \F^{2n}: H_1^c \left(\hat{v}\right)^T=0, \hat{v}_0 \ne 0\right\},
\end{eqnarray*}
and the matrix $H_1^c$ is given by
\[H_1^c:=\left[\begin{array}{cc}
H_0&0\\
\wh_1&H_0\end{array}\right].\]
For $v=(v_0,v_1)$ where $v_0,v_1 \in \F^n$ and $v_0 \ne 0$ such that $H_1^c v^T=0$, that is, $H_0 v_0^T=0$ and $\wh_1 v_0^T+H_0v_1^T=0$. If $v_1=0$, then $H v_0^T=0$, thus $0 \ne v_0 \in \cC$ and $\wt(v_0) \ge d=n-k+1$. If $v_1 \ne 0$, we have $v_0 \in \cC_0$ and $\wt(v_0) \ge d_0=\ga+1$. Noting that $\ga-\de \ge \de-1$, the equation $\wh_1 v_0^T+H_0v_1^T=0$ implies that $\Ga v_1^T=0,v_1 \ne 0$ where
\[\Ga^T=\left[\ul{h}_0^T,\ul{h}_1^T,\ldots,
\ul{h}_{\ga-\de-1}^T\right].\]
It is easy to see that any $\ga-\de$ columns of $\Ga$ are linearly independent. This implies that $\wt(v_1) \ge \ga-\de+1 \ge \de$, hence we have $\wt(v)=\wt(v_0)+\wt(v_1) \ge \ga+1+\de=n-k+1$. This confirms (\ref{2:dineq}) and completes the proof of (iii).
\end{proof}
We remark that (i) of Theorem \ref{3:thm1} can be proved by by using the generalized Reed-Solomon ($\grs$) code as well.


\section{Unit-memory MDS $(q,k,\de)_q$ codes} \label{sec4}

Again let $q \geq 3$ be a prime power. For $n=q$ we find it is most convenient to use the generalized Reed-Solomon ($\grs$) codes of length $q$ to construct unit-memory $(q,k,\de)_q$ codes, though it is quite possible that $f$-cyclic codes may be used for this purpose as well.

Let $\theta$ be a primitive element of $\F$. Let $\zv =\left(0,\theta,\ldots,\theta^{q-1}\right)$ and $\zw=\left(w_0,w_1,\ldots,w_{q-1}\right)$ where $w_i \in \F^*$ for each $i$. For any positive integer $k<q$, the $\grs$ code $\cC=\grs_{k}(\zv,\zw)$ is defined by
\[\grs_{k}(\zv,\zw)=\Big\{\left(w_0f(0),w_1f\left(\theta\right),\ldots,w_{q-1} f\left(\theta^{q-1}\right)\right): f \in \pp_k\Big\},\]
where $\pp_k \subset \F[x]$ is the set of polynomials of degree less than $k$. It is known that $\cC$ is an $[q,k,q-k+1]_q$ MDS linear code and a parity check matrix of $\cC$ is given by
\[H=\left[\begin{array}{ccccc}
u_0&u_1&u_2&\cdots&u_{q-1}\\
0&u_1 \theta &u_2\theta^2& \cdots &u_{q-1}\theta^{q-1}\\
\vdots & \vdots& \vdots& \cdots &\vdots\\
0&u_1 \theta^{q-k-1} &u_2\theta^{2(q-k-1)}& \cdots &u_{q-1}\theta^{(q-1)(q-k-1)}\end{array}\right]_{(q-k) \times q},\]
where $u_i \in \F^*$ can be determined by $\cC$ (see \cite[Theorem 5.3.3]{HV}). Let $\ul{h}_j$ be the $j+1$-th row of $H$. For any positive integers $\ga,\de$ such that $\ga \ge \de$ and $\ga +\de=q-k$, we split the matrix $H$ into 2 disjoint submatrices $H_0$ and $H_1$ such that
$H=\begin{bmatrix}
 H_0 \\
 H_1
\end{bmatrix}$ where
\[H_0^T=\left[\ul{h}_0^T,\ul{h}_1^T,\ldots,\ul{h}_{\ga-1}^T\right], \quad H_1^T=\left[\ul{h}_{\ga+\de-1}^T,\ul{h}_{\ga+\de-2}^T, \ldots,\ul{h}_{\ga}^T\right].\]
Define two $\ga \times q$ matrices $\wh_0,\wh_1$ by
\begin{eqnarray*}
\wh_0=H_0, \quad \wh_1^T=\left[\mathbf{0}_{q \times (\ga-\de)},\ul{h}_{\ga+\de-1}^T,\ul{h}_{\ga+\de-2}^T, \ldots,\ul{h}_{\ga}^T\right].
\end{eqnarray*}
Let $V \in \cF^n$ be the convolutional code defined by the parity-check matrix $G(D)$ given by
\[g(D)=\wh_0+\wh_1D.\]
We prove the following result.

\begin{theorem} \label{4:thm2}
(i). The code $V$ defined above is an $(q,k+\de,\de)_q$ unit-memory MDS convolutional code whenever
\[1 \le k < q, \quad 1 \le \de \le \frac{q-k}{2}\,.\]

(ii). The code $V$ is an $(q,k+\de,\de)_q$ unit-memory MDS convolutional code with a maximal distance profile whenever

\[1 \le k < q, \quad 1 \le \de < \frac{q-k}{2}\,.\]

(iii). The code $V$ defined above is an $(q,k+\de,\de)_q$ unit-memory strongly-MDS convolutional code whenever
\[1 \le k <q, \quad 1 \le \de \le \frac{q-k+1}{3}\,.\]
\end{theorem}

\begin{proof}
It is clear that $V$ is an $(q,k+\de,\de)_q$ code. The arguments of (ii) and (iii) are exactly the same as those of (ii) and (iii) of Theorem \ref{3:thm1} respectively, hence we omit details. We only prove (i), the proof of which is quite different from that of (i) in Theorem \ref{3:thm1}.

To prove (i), it suffices to prove that there exists a sufficiently large integer $j$ such that the $j$-th column distance $d_j^c$ of $\cC$ satisfies $d^c_j \ge n-k+1=\ga+\de+1$, since $\D \ge d_j^c$ for any $j$. Here
\begin{eqnarray*}d_j^c&=&\min\left\{\wt\left(\hat{v}\right): \hat{v}=\left(\hat{v}_0,\ldots,\hat{v}_j\right)\in \F^{(j+1)n}: H_j^c \left(\hat{v}\right)^T=0, \hat{v}_0 \ne 0\right\},
\end{eqnarray*}
where
\[H_j^c=\left[\begin{array}{ccccc}
H_0&&&&\\
\wh_1&H_0 && &\\
& \wh_1 & H_0 &&\\
&& \ddots & \ddots &\\
&&&\wh_1&H_0\end{array}\right].\]
For any $v_i \in \F^n, 0 \le i \le j$ and $v_0 \ne 0$, for the sake of convenience we write each $v_i$ as a column vector. Let $v^T=(v_0^T,\ldots,v_j^T)$ and suppose $H_j^c v=0$. This is equivalent to $H_0v_0=0$ and $\wh_1v_{i-1}+H_0v_i=0$ for each $1 \le i \le j$. Let $\cC_0$ be the linear code defined by the parity-check matrix $H_0$. Obviously $\cC_0$ a $[q,q-\ga,\ga+1]_q$ MDS code and $0 \ne v_0 \in \cC_1$, hence $\wt(v_0) \ge \ga+1$. Our goal is to prove $\wt(v) \ge \ga+\de+1$ when $j$ is sufficiently large. The argument is dividend into several cases, which we will describe as below.

\noindent {\bf Case 1.} If $\# \{1 \le i \le j: v_i \ne 0\} \ge \de$, then $\wt(v)=\sum_i \wt(v_i) \ge \ga+\de+1$, done.

\noindent {\bf Case 2.} If $\# \{1 \le i \le j: v_i \ne 0\} =0$, that is, $v_1=\ldots=v_j=0$, then $Hv_0=0$, hence $v_0 \in \cC$ and $\wt(v_0) \ge d=\gamma+\de+1$, done.


\noindent {\bf Case 3.} If there is some $i$ such that $v_i=0$ but $v_{i+1} \ne 0$, then from $\wh_1v_i+H_0v_{i+1}=H_0v_{i+1}=0$ we have $0 \ne v_{i+1} \in \cC_0$ with $\wt(v_{i+1}) \ge \ga+1$, thus $\wt(v) \ge \wt(v_0)+\wt(v_{i+1}) \ge 2 \ga+2 \ge \ga+\de+1$, done.

It remains to consider the case that there is an integer $i$ in the range $1 \le i< \de$ such that $v_t \ne 0$ for any $t \le i$ and $v_t=0$ for any $t >i$. The equations are $\wh_1v_i+H_0v_{i+1}=\wh_1v_i=0$ and $\wh_1v_{i-1}+H_0v_i=0$.

\noindent {\bf Case 4.} If $\wt(v_i) \ge 2$, from $H_1v_i=0$ and the definition of $H_1$ we see that $\wt(v_i) \ge \de+1$, done.

\noindent {\bf Case 5.} Now suppose $\wt(v_i)=1$. Then $v_i$ must be of the form $v_i^T=(\al,0,\ldots,0)$ for some $\al \in \F^*$. If $\ga > \de$, then it is impossible to have $\wh_1v_{i-1}+H_0v_i=0$. So we shall assume that $\ga=\de$. Thus $\wh_1=H_1$. We may assume that $\ga=\de \ge 2$, because if otherwise then $\wt(v) \ge \wt(v_0)+\wt(v_i) \ge \ga+\de+1$ is trivial.

\noindent {\bf Case 5 i).} If $i \ge 2$, from $H_1v_{i-1}+H_0v_i=H_1v_{i-1}+[\al,0,\ldots,0]^T=0$ and the definition of $H_1$ we obtain $\Ga v_{i-1}=0$ where $\Ga^T=\left[\ul{h}_{\ga+\de-2}^T,\ul{h}_{\ga+\de-3}^T, \ldots,\ul{h}_{\ga}^T\right]$. If $\wt(v_{i-1}) \ge 2$, then $\wt(v_{i-1}) \ge \de$, thus $\wt(v) \ge \wt(v_0)+\wt(v_{i-1})+\wt(v_i) \ge \ga+\de+1$, done. If $\wt(v_{i-1})=1$, then $v_{i-1}$ must be of the form $v_{i-1}^T=(\beta,0,\ldots,0)$ for some $\beta \in \F^*$. However this contradicts to the equation $H_1v_{i-1}+H_0v_i=0$.

\noindent {\bf Case 5 ii).} Finally we consider that case that $i=1$, that is, $v_0 \ne 0, v_1 \ne 0, v_t =0$ for any $t \ge 2$. The equations $H_0v_0=0$ and $\wh_1v_{0}+H_0v_1=0$ where $v_1^T=(\alpha,0,\ldots,0), \alpha \in \F^*$ implies that $\widetilde{\Ga} \hat{v}_0=0$ where $\widetilde{\Ga}=\left[\ul{h}_0,\ldots,\ul{h}_{\ga+\de-2}\right]$, where $\hat{v}_0$ is a rearrangement of entries of $v_0$ and hence they have the same Hamming weight. Thus $\wt(v_0) \ge \ga+\de$ and therefore $\wt(v)=\wt(v_0)+\wt(v_1) \ge \ga+\de+1$, done.

We conclude that in all cases we have $\wt(v) \ge \ga+\de+1$. This proves $d_j^c \ge \ga+\de+1$ for some $j$, and the proof of (i) is now complete.

\end{proof}

\section{Unit-memory MDS $(q+1,k,\de)_q$ codes} \label{sec5}

Let $\theta$ be a primitive element of $\Fqq$ and $\beta:=\theta^{q-1}$. Then $\beta$ is a primitive $(q+1)$-th root of unity in $\Fqq$ and $\beta^q=\beta^{-1}$.

Let $k$ be any positive integer $k$ such that $1 \le k <q+1$. We first consider the case that $k \equiv q \pmod{2}$.

\subsection{Part 1. $k \equiv q \pmod{2}$} Define $\tau:=\frac{q-k}{2}$ and
\[g(x) =\prod_{j=-\tau}^{\tau} \left(x-\beta^j\right).\]
It is clear that $\deg g=2 \tau+1$, $g(x) \in \F[x]$ and $g(x)|\left(x^{q+1}-1\right)$. Let $\cC =(g(x)) \subset \F[x]/\left(x^{q+1}-1\right)$ be the cyclic code of length $q+1$ with generator polynomial $g(x)$. It is known from Lemma \ref{2:lem2} that $\cC$ is a $[q+1,k,2 \tau+2]_q$ MDS code.

Each $v(x)=\sum_{i=0}^{q}v_i x^i \in \cC$ if and only if $g(x)|v(x)$, that is,
\begin{eqnarray} \label{5:gx}v\left(\beta^j\right)=0, \quad \mbox{ for } \forall \, -\tau \le j \le \tau.\end{eqnarray}
Because $v_i \in \F$ and $\beta^q=\beta^{-1}$, this is equivalent to
\begin{eqnarray} \label{5:gx2}v\left(\beta^j\right)=0, \quad \mbox{ for } \forall \, 0 \le j \le \tau.\end{eqnarray}
Denote for each $j$
\[\ul{h}_j=\left[1,\beta^{j},\beta^{2j}, \ldots,\beta^{qj}\right],\]
and define
\[\ol{H}^T=\left[\ul{h}_{-\tau}^T, \ldots,\ul{h}_{-1},\ul{h}_0^T,\ul{h}_1^T,\ldots,\ul{h}_{\tau}^T\right],\]
\begin{eqnarray} \label{5:h2} H^T=\left[\ul{h}_0^T,\ul{h}_1^T,\ldots,
\ul{h}_{\tau}^T\right].\end{eqnarray}
(\ref{5:gx}) and (\ref{5:gx2}) imply that for any $v=(v_0,\ldots,v_q) \in \cC \subset \F^{q+1}$ if and only if $\ol{H}v^T=0$ if and only if $Hv^T=0$.

Each vector $\ul{h}_j (1 \le j \le \tau)$ which is originally defined over $\Fqq$ can be identified with two rows of length $q+1$ defined over $\F$. More precisely, let $\{1,e\}$ be a basis of $\Fqq$ over $\F$, then $\ul{h}_j=\ul{h}_{j,1}+e \ul{h}_{j,2}$ where $\ul{h}_{j,i}, i=1,2$ are defined over $\F$. We identify $\ul{h}_j$ with $\left[\begin{array}{c}
\ul{h}_{j,1}\\
\ul{h}_{j,2}\end{array}\right]$, which is a $2 \times (q+1)$ matrix over $\F$. Under this identification $H$ becomes a $(2\tau+1) \times (q+1)$ matrix over $\F$, and $v=(v_0,\ldots,v_q) \in \cC$ if and only if $Hv^T=0$, hence $H$ can be regarded as a parity-check matrix of $\cC$. In what follows each $\ul{h}_j$, $j \ge 1$ shall be thought of as either one row over $\Fqq$ or equivalently two rows over $\F$.

To construct convolutional codes from $H$, we split the matrix $H$ into 2 disjoint submatrices $H_0$ and $H_1$ such that
$H=\begin{bmatrix}
 H_0 \\
 H_1
\end{bmatrix}$. There are two different ways to do that.

\noindent {\bf Construction one.} Let $q \geq 5$. For any positive integers $\ga,\de$ such that $\ga > \de$ and $\ga +\de=\tau+1$, define $H_0,H_1$ as
\[H_0^T=\left[\ul{h}_0^T,\ul{h}_1^T,\ldots,\ul{h}_{\ga-1}^T\right], \quad H_1^T=\left[\ul{h}_{\ga}^T,\ul{h}_{\ga+1}^T, \ldots,\ul{h}_{\ga+\de-1}^T\right].\]
Considered as matrices over $\F$, $H_0,H_1$ are of size $(2 \ga-1) \times (q+1)$ and $2 \de \times (q+1)$ respectively. It shall be noted that any $(2 \ga-1)$ columns of $H_0$ are linearly independent, because for any $v=(v_1,\ldots,v_q) \in \F^{q+1}$, $H_0v^T=0$ if and only if $\ol{H}_0v^T=0$ where
\[\ol{H}_0^T=\left[\ul{h}_{-(\ga-1)}^T,\ldots,\ul{h}_{-1}^T,\ul{h}_0^T,
\ul{h}_1^T,\ldots,\ul{h}_{\ga-1}^T\right].\]

Define two $(2\ga-1) \times (q+1)$ matrices $\wh_0,\wh_1$ by
\begin{eqnarray*}
\wh_0=H_0, \quad \wh_1^T=\left[\mathbf{0}_{(q+1) \times (2\ga-1-2\de)},\ul{h}_{\ga}^T,\ul{h}_{\ga+1}^T, \ldots,\ul{h}_{\ga+\de-1}^T\right].
\end{eqnarray*}
Let $V \in \cF^{q+1}$ be the convolutional code defined by the parity-check matrix $G(D)$ given by
\[g(D)=\wh_0+\wh_1D.\]
We prove the following result.

\begin{theorem} \label{5:thm3}
The code $V$ defined above is a $(q+1,k+2\de,2\de)_q$ unit-memory strongly-MDS convolutional code whenever
\[k \equiv q \pmod{2}, \quad 1 \le k <q, \quad 1 \le \de \le \frac{q-k+2}{6}\,.\]
In this case $V$ also has a maximal distance profile.
\end{theorem}

\begin{proof}
First it is clear that $V$ is an $(q+1,k+2\de,2\de)_q$ convolutional code.

Since $M=\lfloor 2\de/(k+2\de)\rfloor+\lceil 2\de/(q+1-k-2\de)\rceil=1$ and $d_1^c \le \D \le d=q+2-k=$ by the generalized Singleton bound (\ref{1:mds}), it suffices to prove that
\begin{eqnarray}\label{5:dineq} d_1^c \ge d=q+2-k,\end{eqnarray}
where $d_1^c$ is the $1$-st column distance of $\cC$ which is given by
\begin{eqnarray*}d_1^c
&=& \min\left\{\wt\left(\hat{v}\right): \hat{v}=\left(\hat{v}_0,\hat{v}_1\right)\in \F^{2(q+1)}: H_1^c \left(\hat{v}\right)^T=0, \hat{v}_0 \ne 0\right\}.
\end{eqnarray*}
The matrix $H_1^c$ is given by
\[H_1^c:=\left[\begin{array}{cc}
H_0&0\\
\wh_1&H_0\end{array}\right].\]
For $v=(v_0,v_1)$ where $v_0,v_1 \in \F^{q+1}$ and $v_0 \ne 0$ such that $H_1^c v^T=0$, that is, $H_0 v_0^T=0$ and $\wh_1 v_0^T+H_0v_1^T=0$. If $v_1=0$, then $H v_0^T=0$, thus $0 \ne v_0 \in \cC$ and $\wt(v_0) \ge d=q+2-k$. If $v_1 \ne 0$, we have $v_0 \in \cC_0$ and $\wt(v_0) \ge d_0=2\ga$. Noting that $\ga>\de$, the equation $\wh_1 v_0^T+H_0v_1^T=0$ implies that $\Ga v_1^T=0$ where
\[\Ga^T=\left[\ul{h}_0^T,\ul{h}_1^T,\ldots,\ul{h}_{\ga-\de-1}^T\right].\]
Define
\[\ol{\Ga}^T=\left[\ul{h}_{-(\ga-\de-1)}^T,\ldots,\ul{h}_0^T,\ul{h}_1^T,\ldots,\ul{h}_{\ga-\de-1}^T\right].\]
It is easy to see that $\Ga v_1^T=0$ if and only if $\ol{\Ga} v_1^T=0$, since $v_1 \ne 0$, it implies that $\wt(v_1) \ge 2(\ga-\de) \ge 2 \de$, hence we still have $\wt(v)=\wt(v_0)+\wt(v_1) \ge 2\ga+2\de=q+2-k$. This confirms (\ref{5:dineq}) and completes the proof of Theorem \ref{5:thm3}.
\end{proof}

\noindent {\bf Construction two.} We assume that $q \geq 4$ is even. Let $\he$ be the submatrix of $H$ which consists of the $\ul{h}_j$ for all even $j$ where $0 \le j \le \tau$, and let $\ho$ be the submatrix of $H$ which consists of the $\ul{h}_j$ for all odd $j$ where $0 \le j \le \tau$. Let $r$ and $s$ the number of even $j$'s and the number of odd $j$'s respectively in the range $0 \le j \le \tau$. Then $r+s=\tau+1$. Let $\cCo$ and $\cCe$ be the linear codes obtained by $\ho,\he$ as parity-check matrices respectively. It is easy to see that both $\cCo,\cCe$ are MDS linear codes when $q$ is even (see Lemma \ref{2:lem2}). This is because, for example, for any $v=(v_0,\ldots,v_q) \in \F^{q+1}$, $H_ev^T=0$ if and only if $\ol{H}_{\mathrm{e}} v^T=0$ where $\ol{H}_{\mathrm{e}}$ is given by
\[\ol{H}_{\mathrm{e}}^T=\left[\ldots,\ul{h}_{-2}^T,\ul{h}_{0}^T,
\ul{h}_{2}^T,\ldots\right].\]
Now it is clear that any $2r-1$ columns of $\he$ are linearly independent. This argument also applies to $\ho$.

We consider $\he,\ho$ as matrices over $\F$. If $2 r-1>2s$, define $\wh_0=\he$ and define $\wh_1$ to be the matrix obtained from $\ho$ by adding zero rows; if $2r-1<2s$, define $\wh_0=\ho$ and define $\wh_1$ to be the matrix obtained from $\he$ by adding zero rows. Define
\[G(D)=\wh_0+\wh_1D, \]
and let $V \in \cF^{q+1}$ be the convolutional code defined by the parity-check matrix $G(D)$. We have the following result.

\begin{theorem} \label{51:thm4}
If $q \geq 4$ is even, the code $V$ defined above is a $\left(q+1,q-\tau,\tau\right)_q$ MDS convolutional code for any $\tau$ such that
\[1 \le \tau \le \frac{q-1}{2}. \]
\end{theorem}
\begin{proof}
Lemma \ref{2:lem3} implies directly that $V$ is an MDS convolutional code of length $q+1$. As for other parameters of $V$, if $\tau$ is even, then $r=\frac{\tau}{2}+1,s=\frac{\tau}{2}$, and $2r-1>2s$, then $V$ has degree $2s=\tau$ and dimension $q+1-(2r-1)=q-\tau$. If $\tau$ is odd, then $r=s=\frac{\tau+1}{2}$, $2r-1<2s$, hence $V$ has degree $2r-1=\tau$ and dimension $q+1-2s=q-\tau$. This concludes the proof of Theorem \ref{51:thm4}.
\end{proof}

\subsection{Part 2. $k \equiv q+1 \pmod{2}$}

Again let $q \geq 5$. Define $\tau:=\frac{q-k-1}{2}$ and
\[g(x) =\prod_{j=-\tau}^{\tau+1} \left(x-\theta\beta^j\right).\]
It is clear that $\deg g=2 \tau+2=q+1-k$, $g(x) \in \F[x]$ and $g(x)|\left(x^{q+1}-\theta^{q+1}\right)$. Let $\cC =(g(x)) \subset \F[x]/\left(x^{q+1}-\theta^{q+1}\right)$ be the constacyclic code of length $q+1$ with generator polynomial $g(x)$. It is known from Lemma \ref{2:lem2} that $\cC$ is a $[q+1,k,2 \tau+3]_q$ MDS code.

Each $v(x)=\sum_{i=0}^{q}v_i x^i \in \cC$ if and only if $g(x)|v(x)$, that is,
\begin{eqnarray} \label{51:gx}v\left(\theta\beta^j\right)=0, \quad \mbox{ for } \forall \, -\tau \le j \le \tau+1.\end{eqnarray}
Because $v_i \in \F$ and $(\theta\beta^j)^q=\theta\beta^{1-j}$, this is equivalent to
\begin{eqnarray} \label{51:gx2}v\left(\theta\beta^j\right)=0, \quad \mbox{ for } \forall \, 1 \le j \le \tau+1.\end{eqnarray}
Denote for each $j$
\[\ul{h}_j=\left[1,\theta\beta^{j},(\theta\beta^j)^{2}, \ldots,(\theta\beta^j)^{q}\right],\]
and define
\[\ol{H}^T=\left[\ul{h}_{-\tau}^T, \ldots,\ul{h}_{-1},\ul{h}_0^T,\ul{h}_1^T,\ldots,\ul{h}_{\tau+1}^T\right],\]
\[H^T=\left[\ul{h}_1^T,\ldots,\ul{h}_{\tau+1}^T\right].\]
(\ref{51:gx}) and (\ref{51:gx2}) imply that for any $v=(v_0,\ldots,v_q) \in \cC \subset \F^{q+1}$ if and only if $\ol{H}v^T=0$ if and only if $Hv^T=0$.

Similar as before, we identify each $\ul{h}_j$ ($1 \le j \le \tau+1$) with $\left[\begin{array}{c}
\ul{h}_{j,1}\\
\ul{h}_{j,2}\end{array}\right]$, which is a $2 \times (q+1)$ matrix over $\F$. Under this identification $H$ becomes a $(2 \tau+2) \times (q+1)$ matrix over $\F$, and $v=(v_0,\ldots,v_q) \in \cC$ if and only if $Hv^T=0$, hence $H$ can be regarded as a parity-check matrix of $\cC$. In what follows each $\ul{h}_j$, $j \ge 1$ shall be thought of as either one row over $\Fqq$ or equivalently two rows over $\F$.

For any positive integers $\ga,\de$ such that $\ga \ge \de$ and $\ga +\de=\tau+1$, we split the matrix $H$ into 2 disjoint submatrices $H_0$ and $H_1$ such that $H=\left[\begin{array}{c}
H_0\\
H_1\end{array}\right]$, where $H_0$ and $H_1$ are given by
\[H_0^T=\left[\ul{h}_1^T,\ldots,\ul{h}_{\ga}^T\right], \quad H_1^T=\left[\ul{h}_{\ga+1}^T,\ul{h}_{\ga+2}^T, \ldots,\ul{h}_{\ga+\de}^T\right].\]
Considered as matrices over $\F$, $H_0,H_1$ are of size $2 \ga \times (q+1)$ and $2 \de \times (q+1)$ respectively. It shall be noted that any $2 \ga$ columns of $H_0$ are linearly independent.

Define two $2\ga \times (q+1)$ matrices $\wh_0,\wh_1$ by
\begin{eqnarray*}
\wh_0=H_0, \quad \wh_1^T=\left[\mathbf{0}_{(q+1) \times (2\ga-2\de)},\ul{h}_{\ga+1}^T,\ul{h}_{\ga+2}^T, \ldots,\ul{h}_{\ga+\de}^T\right].
\end{eqnarray*}
Let $V \in \cF^{q+1}$ be the convolutional code defined by the parity-check matrix $G(D)$ given by
\[g(D)=\wh_0+\wh_1D.\]
We prove the following result.

\begin{theorem} \label{5:thm4}
The code $V$ defined above is a $(q+1,k+2\de,2\de)_q$ unit-memory strongly-MDS convolutional code whenever
\[k \equiv q+1 \pmod{2}, \quad 1 \le k <q, \quad 1 \le \de \le \frac{q-k+1}{6}\,.\]
In this case $V$ also has a maximal distance profile.
\end{theorem}

\begin{proof}
The proof is very similar to that of Theorem \ref{5:thm3}. For the sake of completeness, we provide all details here.

First it is clear that $V$ is an $(q+1,k+2\de,2\de)_q$ convolutional code.

Since $M=\lfloor 2\de/(k+2\de)\rfloor+\lceil 2\de/(q+1-k-2\de)\rceil=1$ and $d_1^c \le \D \le d=q+2-k$ by the generalized Singleton bound, it suffices to prove that
\begin{eqnarray}\label{51:dineq} d_1^c \ge d=q+2-k,\end{eqnarray}
where $d_1^c$ is the $1$-st column distance of $\cC$ which is given by
\begin{eqnarray*}d_1^c
&=& \min\left\{\wt\left(\hat{v}\right): \hat{v}=\left(\hat{v}_0,\hat{v}_1\right)\in \F^{2n}: H_1^c \left(\hat{v}\right)^T=0, \hat{v}_0 \ne 0\right\},
\end{eqnarray*}
and the matrix $H_1^c$ is given by
\[H_1^c:=\left[\begin{array}{cc}
H_0&0\\
\wh_1&H_0\end{array}\right].\]
For $v=(v_0,v_1)$ where $v_0,v_1 \in \F^{q+1}$ and $v_0 \ne 0$ such that $H_1^c v^T=0$, that is, $H_0 v_0^T=0$ and $\wh_1 v_0^T+H_0v_1^T=0$. If $v_1=0$, then $H v_0^T=0$, thus $0 \ne v_0 \in \cC$ and $\wt(v_0) \ge d=q+2-k$. If $v_1 \ne 0$, we have $v_0 \in \cC_0$ and $\wt(v_0) \ge d_0=2\ga+1$. Noting that $\ga>\de$, the equation $\wh_1 v_0^T+H_0v_1^T=0$ implies that $\Ga v_1^T=0$ where
\[\Ga^T=\left[\ul{h}_1^T,\ldots,\ul{h}_{\ga-\de}^T\right].\]
Define
\[\ol{\Ga}^T=\left[\ul{h}_{1-(\ga-\de)}^T,\ldots,\ul{h}_0^T,
\ul{h}_1^T,\ldots,
\ul{h}_{\ga-\de}^T\right].\]
It is easy to see that $\Ga v_1^T=0$ if and only if $\ol{\Ga} v_1^T=0$, since $v_1 \ne 0$, it implies that $\wt(v_1) \ge 2(\ga-\de)+1 \ge 2 \de$, hence we still have $\wt(v)=\wt(v_0)+\wt(v_1) \ge 2\ga+2\de+1=q+2-k$. This confirms (\ref{51:dineq}) and completes the proof of Theorem \ref{5:thm4}.
\end{proof}

\section{Examples of MDS convolutional codes over $\mathbb{F}_8$} \label{sec6}

In this section, we present some explicit examples of MDS and strongly-MDS convolutional codes over $\mathbb{F}_8$. Hence $q=8$. Let $\theta$ be a root of the irreducible polynomial $x^3+x+1$ over $\mathbb{F}_2$. We know that $\theta$ is a primitive element of $\mathbb{F}_8$, that is, $\theta$ has order $8-1=7$. The elements of $\mathbb{F}_8$ are: $0,1, \theta, \theta^2, \theta^3=1+\theta, \theta^4=\theta+\theta^2, \theta^5=1+\theta+\theta^2, \theta^6=1+\theta^2$. We consider the following cases:

\subsection{Case 1: $n=q-1=7$}


From Theorem \ref{3:thm1}, we consider cyclic codes $\cC$ of length 7 over $\mathbb{F}_8$, that is, $\cC =(g(x)) \subset \mathbb{F}_8[x]/\left(x^7-1\right)$.

\noindent{\bf{Example 1: MDS $(7,4,2)_8$}}

Here $k=2, \de=2$. Consider the cyclic code $\cC$ of length 7 over $\mathbb{F}_8$ with generator polynomial $g(x)=\prod_{i=0}^4 (x-\theta^i)$. Note that $\deg g=5$, so $\cC$ is a $[7, 2, 6]_8$ MDS linear code by Lemma 1.

A parity-check matrix $H$ of $\cC$ is given by
$$H = \begin{bmatrix}
1 & 1 & 1 & 1 & 1 & 1 & 1\\
1 & \theta & \theta^2 & \theta^3 & \theta^4 & \theta^5 & \theta^6\\
1 & \theta^2 & \theta^4 & \theta^6 & \theta^8 & \theta^{10} & \theta^{12}\\
1 & \theta^3 & \theta^6 & \theta^9 & \theta^{12} & \theta^{15} & \theta^{18}\\
1 & \theta^4 & \theta^8 & \theta^{12} & \theta^{16} & \theta^{20} & \theta^{24}
\end{bmatrix}$$

$$=\begin{bmatrix}
1 & 1 & 1 & 1 & 1 & 1 & 1\\
1 & \theta & \theta^2 & 1+\theta & \theta+\theta^2 & 1+\theta+\theta^2 & 1+\theta^2\\
1 & \theta^2 & \theta+\theta^2 & 1+\theta^2 & \theta & 1+\theta & 1+\theta+\theta^2\\
1 & 1+\theta & 1+\theta^2 & \theta^2 & 1+\theta+\theta^2 & \theta & \theta+\theta^2\\
1 & \theta+\theta^2 & \theta & 1+\theta+\theta^2 & \theta^2 & 1+\theta^2 & 1+\theta
\end{bmatrix}$$

Now let $H_{0}$ be the first $n-k-\de=3$ rows of $H$ and $H_{1}$ be the rest $\delta=2$ rows. Then consider the convolutional code $V$ with parity-check matrix

$$G(D)=H_{0}+\widetilde{H}_{1}D$$

\noindent{\footnotesize
\arraycolsep=3pt
\medmuskip = 1mu
$$=\begin{bmatrix}
1 & 1 & 1 & 1 & 1 & 1 & 1\\
1+D & \theta+(1+\theta)D & \theta^2+(1+\theta^2)D & (1+\theta)+\theta^2D & (\theta+\theta^2)+(1+\theta+\theta^2)D & (1+\theta+\theta^2)+\theta D & (1+\theta^2)+(\theta+\theta^2)D\\
1+D & \theta^2+(\theta+\theta^2)D & (\theta+\theta^2)+\theta D & (1+\theta^2)+(1+\theta+\theta^2)D & \theta+\theta^2D & (1+\theta)+(1+\theta^2)D & (1+\theta+\theta^2)+(1+\theta)D
\end{bmatrix}$$
}

By Theorem \ref{3:thm1} (i), $V$ is a unit-memory MDS convolutional code with parameters $(7, 4, 2)_8$. The free distance of $V$ is $6$. Noting by (ii) and (iii) of Theorem \ref{3:thm1}, $V$ has a maximal distance profile and is also strongly-MDS.

\noindent{\bf{Example 2: MDS $(7,3,2)_8$}}

Here $k=1, \de=2$. Consider the cyclic code $\cC$ of length 7 over $\mathbb{F}_8$ with generator polynomial $g(x)=\prod_{i=0}^5 (x-\theta^i)$. Note that $\deg g=6$, so $\cC$ is a $[7, 1, 7]_8$ MDS linear code by Lemma 1.

A parity-check matrix $H$ of $\cC$ is given by
$$H = \begin{bmatrix}
1 & 1 & 1 & 1 & 1 & 1 & 1\\
1 & \theta & \theta^2 & \theta^3 & \theta^4 & \theta^5 & \theta^6\\
1 & \theta^2 & \theta^4 & \theta^6 & \theta^8 & \theta^{10} & \theta^{12}\\
1 & \theta^3 & \theta^6 & \theta^9 & \theta^{12} & \theta^{15} & \theta^{18}\\
1 & \theta^4 & \theta^8 & \theta^{12} & \theta^{16} & \theta^{20} & \theta^{24}\\
1 & \theta^5 & \theta^{10} & \theta^{15} & \theta^{20} & \theta^{25} & \theta^{30}
\end{bmatrix}$$
$$=\begin{bmatrix}
1 & 1 & 1 & 1 & 1 & 1 & 1\\
1 & \theta & \theta^2 & 1+\theta & \theta+\theta^2 & 1+\theta+\theta^2 & 1+\theta^2\\
1 & \theta^2 & \theta+\theta^2 & 1+\theta^2 & \theta & 1+\theta & 1+\theta+\theta^2\\
1 & 1+\theta & 1+\theta^2 & \theta^2 & 1+\theta+\theta^2 & \theta & \theta+\theta^2\\
1 & \theta+\theta^2 & \theta & 1+\theta+\theta^2 & \theta^2 & 1+\theta^2 & 1+\theta\\
1 & 1+\theta+\theta^2 & 1+\theta & \theta & 1+\theta^2 & \theta+\theta^2 & \theta^2
\end{bmatrix}$$

Now let $H_{0}$ be the first $n-k-\de=4$ rows of $H$ and $H_{1}$ be the rest $\de=2$ rows. Then consider the convolutional code $V$ with parity-check matrix

$$G(D)=H_{0}+\widetilde{H}_{1}D$$
\noindent{\footnotesize
\arraycolsep=3pt
\medmuskip = 1mu
$$=\begin{bmatrix}
1 & 1 & 1 & 1 & 1 & 1 & 1\\
1 & \theta & \theta^2 & 1+\theta & \theta+\theta^2 & 1+\theta+\theta^2 & 1+\theta^2\\
1+D & \theta^2+(\theta+\theta^2)D & (\theta+\theta^2)+\theta D & (1+\theta^2)+(1+\theta+\theta^2)D & \theta+\theta^2D & (1+\theta)+(1+\theta^2)D & (1+\theta+\theta^2)+(1+\theta)D\\
1+D & (1+\theta)+(1+\theta+\theta^2)D & (1+\theta^2)+(1+\theta)D & \theta^2+\theta D & (1+\theta+\theta^2)+(1+\theta^2)D & \theta+(\theta+\theta^2)D & (\theta+\theta^2)+\theta^2D
\end{bmatrix}$$
}

By (i) of Theorem \ref{3:thm1} (i), $V$ is a unit-memory MDS convolutional code with parameters $(7, 3, 2)_8$. The free distance of $V$ is $7$. Noting by (ii) and (iii) of Theorem \ref{3:thm1}, $V$ has a maximal distance profile and is also strongly-MDS.

\noindent{\bf{Example 3: MDS $(7,4,3)_8$}}

Here $k=1, \de=3$. We use the cyclic code $\cC$ defined in {\bf Example 2}, and we let $H_{0}$ be the first $n-k-\de=3$ rows of $H$ and $H_{1}$ be the rest $\de=3$ rows of $H$ where $H$ is defined in {\bf Example 2}. Then consider the convolutional code $V$ with parity-check matrix

$$G(D)=H_{0}+\widetilde{H}_{1}D$$
\noindent{\footnotesize
\arraycolsep=3pt
\medmuskip = 1mu
$$=\begin{bmatrix}
1+D & 1+(1+\theta)D & 1+(1+\theta^2)D & 1+\theta^2D & 1+(1+\theta+\theta^2)D & 1+\theta D & 1+(\theta+\theta^2)D\\
1+D & \theta+(\theta+\theta^2)D & \theta^2+\theta D & (1+\theta)+(1+\theta+\theta^2)D & (\theta+\theta^2)+\theta^2D & (1+\theta+\theta^2)+(1+\theta^2)D & (1+\theta^2)+(1+\theta)D\\
1+D & \theta^2+(1+\theta+\theta^2)D & (\theta+\theta^2)+(1+\theta)D & (1+\theta^2)+\theta D & \theta+(1+\theta^2)D & (1+\theta)+(\theta+\theta^2)D & (1+\theta+\theta^2)+\theta^2D
\end{bmatrix}$$
}
By Theorem \ref{3:thm1} (i), $V$ is a unit-memory MDS convolutional code with parameters $(7, 4, 3)_8$. The free distance of $V$ is $7$.

\subsection{Case 2: $n=q=8$}

From Theorem \ref{4:thm2}, we use the $\grs$ codes $\cC'_k=\grs_k(\zv, \mathbf{1})$, where $\zv=(0, \theta, \theta^2, \theta^3, \theta^4, \theta^5, \theta^6, \theta^7=1)$ and $\mathbf{1}$ is the all-one vector of length 8. Note that ${\cC'}_k^{\bot}=\grs_{8-k}(\zv, \mathbf{1})$ and $\cC'_k$ is MDS with parameters $[8, k, 9-k]_8$ for $1 \le k \le 8$.

\noindent{\bf{Example 4: MDS $(8,4,2)_8$}}

Here $k=2, \de=2$. A parity-check matrix $H$ of $\cC'_2$ is given by

$$H=\begin{bmatrix}
1 & 1 & 1 & 1 & 1 & 1 & 1 & 1\\
0 & \theta & \theta^2 & \theta^3 & \theta^4 & \theta^5 & \theta^6 & 1\\
0 & \theta^2 & \theta^4 & \theta^6 & \theta^8 & \theta^{10} & \theta^{12} & 1\\
0 & \theta^3 & \theta^6 & \theta^{9} & \theta^{12} & \theta^{15} & \theta^{18} & 1\\
0 & \theta^4 & \theta^8 & \theta^{12} & \theta^{16} & \theta^{20} & \theta^{24} & 1\\
0 & \theta^5 & \theta^{10} & \theta^{15} & \theta^{20} & \theta^{25} & \theta^{30} & 1
\end{bmatrix}$$
$$=\begin{bmatrix}
1 & 1 & 1 & 1 & 1 & 1 & 1 & 1\\
0 & \theta & \theta^2 & 1+\theta & \theta+\theta^2 & 1+\theta+\theta^2 & 1+\theta^2 & 1\\
0 & \theta^2 & \theta+\theta^2 & 1+\theta^2 & \theta & 1+\theta & 1+\theta+\theta^2 & 1\\
0 & 1+\theta & 1+\theta^2 & \theta^2 & 1+\theta+\theta^2 & \theta & \theta+\theta^2 & 1\\
0 & \theta+\theta^2 & \theta & 1+\theta+\theta^2 & \theta^2 & 1+\theta^2 & 1+\theta & 1\\
0 & 1+\theta+\theta^2 & 1+\theta & \theta & 1+\theta^2 & \theta+\theta^2 & \theta^2 & 1
\end{bmatrix}$$

Now let $H_{0}$ be the first $q-k-\de=4$ rows of $H$ and $H_{1}$ be the rest $\de=2$ rows but swapped in position. Then consider the convolutional code $V$ with parity-check matrix

$$G(D)=H_{0}+\widetilde{H}_{1}D$$
\noindent{\footnotesize
\arraycolsep=3pt
\medmuskip = 1mu
$$=\begin{bmatrix}
1 & 1 & 1 & 1 & 1 & 1 & 1 & 1\\
0 & \theta & \theta^2 & 1+\theta & \theta+\theta^2 & 1+\theta+\theta^2 & 1+\theta^2 & 1\\
0 & \theta^2+(1+\theta+\theta^2)D & (\theta+\theta^2)+(1+\theta)D & (1+\theta^2)+\theta D & \theta+(1+\theta^2)D & (1+\theta)+(\theta+\theta^2)D & (1+\theta+\theta^2)+\theta^2D & 1+D\\
0 & (1+\theta)+(\theta+\theta^2)D & (1+\theta^2)+\theta D & \theta^2+(1+\theta+\theta^2)D & (1+\theta+\theta^2)+\theta^2D & \theta+(1+\theta^2)D & (\theta+\theta^2)+(1+\theta)D & 1+D
\end{bmatrix}$$
}
By Theorem \ref{4:thm2} (i), $V$ is a unit-memory MDS convolutional code with parameters $(8, 4, 2)_8$. The free distance of $V$ is $7$. Noting that by (ii) and (iii) of Theorem \ref{4:thm2}, $V$ has a maximal distance profile and is also strongly-MDS.

\noindent{\bf{Example 5: MDS $(8,5,3)_8$}}

Here $k=2, \de=3$. Let $H_{0}$ be the first $q-k-\de=3$ rows of $H$ and $H_{1}$ be the rest $\de=3$ rows of $H$ but reversed in position where $H$ is defined in {\bf Example 4}. Then consider the convolutional code $V$ with parity-check matrix

$$G(D)=H_{0}+\widetilde{H}_{1}D$$
\noindent{\footnotesize
\arraycolsep=3pt
\medmuskip = 1mu
$$=\begin{bmatrix}
1 & 1+(1+\theta+\theta^2)D & 1+(1+\theta)D & 1+\theta D & 1+(1+\theta^2)D & 1+(\theta+\theta^2)D & 1+\theta^2D & 1+D\\
0 & \theta+(\theta+\theta^2)D & \theta^2+\theta D & (1+\theta)+(1+\theta+\theta^2)D & (\theta+\theta^2)+\theta^2D & (1+\theta+\theta^2)+(1+\theta^2)D & (1+\theta^2)+(1+\theta)D & 1+D\\
0 & \theta^2+(1+\theta)D & (\theta+\theta^2)+(1+\theta^2)D & (1+\theta^2)+\theta^2D & \theta+(1+\theta+\theta^2)D & (1+\theta)+\theta D & (1+\theta+\theta^2)+(\theta+\theta^2)D & 1+D
\end{bmatrix}$$
}

By Theorem \ref{4:thm2} (i), $V$ is a unit-memory MDS convolutional code with parameters $(8, 5, 3)_8$. The free distance of $V$ is $7$.

\noindent{\bf{Example 6: MDS $(8,3,2)_8$}}

Here $k=1, \de=2$. A parity-check matrix $H$ of $\cC'_1$ is given by

$$H=\begin{bmatrix}
1 & 1 & 1 & 1 & 1 & 1 & 1 & 1\\
0 & \theta & \theta^2 & \theta^3 & \theta^4 & \theta^5 & \theta^6 & 1\\
0 & \theta^2 & \theta^4 & \theta^6 & \theta^8 & \theta^{10} & \theta^{12} & 1\\
0 & \theta^3 & \theta^6 & \theta^{9} & \theta^{12} & \theta^{15} & \theta^{18} & 1\\
0 & \theta^4 & \theta^8 & \theta^{12} & \theta^{16} & \theta^{20} & \theta^{24} & 1\\
0 & \theta^5 & \theta^{10} & \theta^{15} & \theta^{20} & \theta^{25} & \theta^{30} & 1\\
0 & \theta^6 & \theta^{12} & \theta^{18} & \theta^{24} & \theta^{30} & \theta^{36} & 1
\end{bmatrix}$$
$$=\begin{bmatrix}
1 & 1 & 1 & 1 & 1 & 1 & 1 & 1\\
0 & \theta & \theta^2 & 1+\theta & \theta+\theta^2 & 1+\theta+\theta^2 & 1+\theta^2 & 1\\
0 & \theta^2 & \theta+\theta^2 & 1+\theta^2 & \theta & 1+\theta & 1+\theta+\theta^2 & 1\\
0 & 1+\theta & 1+\theta^2 & \theta^2 & 1+\theta+\theta^2 & \theta & \theta+\theta^2 & 1\\
0 & \theta+\theta^2 & \theta & 1+\theta+\theta^2 & \theta^2 & 1+\theta^2 & 1+\theta & 1\\
0 & 1+\theta+\theta^2 & 1+\theta & \theta & 1+\theta^2 & \theta+\theta^2 & \theta^2 & 1\\
0 & 1+\theta^2 & 1+\theta+\theta^2 & \theta+\theta^2 & 1+\theta & \theta^2 & \theta & 1
\end{bmatrix}$$

Now let $H_{0}$ be the first $q-k-\de=5$ rows of $H$ and $H_{1}$ be the rest $\de=2$ rows but swapped in position. Then consider the convolutional code $V$ with parity-check matrix

$$G(D)=H_{0}+\widetilde{H}_{1}D$$
\noindent{\footnotesize
\arraycolsep=3pt
\medmuskip = 1mu
$$=\begin{bmatrix}
1 & 1 & 1 & 1 & 1 & 1 & 1 & 1\\
0 & \theta & \theta^2 & 1+\theta & \theta+\theta^2 & 1+\theta+\theta^2 & 1+\theta^2 & 1\\
0 & \theta^2 & \theta+\theta^2 & 1+\theta^2 & \theta & 1+\theta & 1+\theta+\theta^2 & 1\\
0 & (1+\theta)+(1+\theta^2)D & (1+\theta^2)+(1+\theta+\theta^2)D & \theta^2+(\theta+\theta^2)D & (1+\theta+\theta^2)+(1+\theta)D & \theta+\theta^2D & (\theta+\theta^2)+\theta D & 1+D\\
0 & (\theta+\theta^2)+(1+\theta+\theta^2)D & \theta+(1+\theta)D & (1+\theta+\theta^2)+\theta D & \theta^2+(1+\theta^2)D & (1+\theta^2)+(\theta+\theta^2)D & (1+\theta)+\theta^2D & 1+D
\end{bmatrix}$$
}
By Theorem \ref{4:thm2} (i), $V$ is a unit-memory MDS convolutional code with parameters $(8, 3, 2)_8$. The free distance of $V$ is $8$. By (ii) and (iii) of Theorem \ref{4:thm2}, $V$ has a maximal distance profile and is also strongly-MDS.

\noindent{\bf{Example 7: MDS $(8,4,3)_8$}}

Here $k=1, \de=3$. Let $H_{0}$ be the first $q-k-\de=4$ rows of $H$ and $H_{1}$ be the rest $\de=3$ rows of $H$ but reversed in position where $H$ is defined in {\bf Example 6}. Then consider the convolutional code $V$ with parity-check matrix

$$G(D)=H_{0}+\widetilde{H}_{1}D$$
\noindent{\footnotesize
\arraycolsep=3pt
\medmuskip = 1mu
$$=\begin{bmatrix}
1 & 1 & 1 & 1 & 1 & 1 & 1 & 1\\
0 & \theta+(1+\theta^2)D & \theta^2+(1+\theta+\theta^2)D & (1+\theta)+(\theta+\theta^2)D & (\theta+\theta^2)+(1+\theta)D & (1+\theta+\theta^2)+\theta^2D & (1+\theta^2)+\theta D & 1+D\\
0 & \theta^2+(1+\theta+\theta^2)D & (\theta+\theta^2)+(1+\theta)D & (1+\theta^2)+\theta D & \theta+(1+\theta^2)D & (1+\theta)+(\theta+\theta^2)D & (1+\theta+\theta^2)+\theta^2D & 1+D\\
0 & (1+\theta)+(\theta+\theta^2)D & (1+\theta^2)+\theta D & \theta^2+(1+\theta+\theta^2)D & (1+\theta+\theta^2)+\theta^2D & \theta+(1+\theta^2)D & (\theta+\theta^2)+(1+\theta)D & 1+D
\end{bmatrix}$$
}
By Theorem \ref{4:thm2} (i), $V$ is a unit-memory MDS convolutional code with parameters $(8, 4, 3)_8$. The free distance of $V$ is $8$. By (ii) of Theorem \ref{4:thm2}, $V$ has a maximal distance profile.

\subsection{Case 3: $n=q+1=9$}

We shall use either {\bf Construction one} or {\bf two} in part 1 of Section \ref{sec5} when $k \equiv q \pmod{2}$ and part 2 of Section \ref{sec5} when $k \equiv q+1 \pmod{2}$. We still use the same $\theta$ as in Cases 1 and 2, and let $\ga$ be a root of the polynomial $x^2+x+1 \in \mathbb{F}_2[x]$. Then $\theta$ has degree 3 over $\mathbb{F}_2$, while $\ga$ has degree 2 over $\mathbb{F}_2$. Hence the field $\mathbb{F}_2(\theta, \ga)$ has degree 6 over $\mathbb{F}_2$, i.e. it is isomorphic to $\mathbb{F}_{64}$.

Let $\beta=\theta(\theta+\ga) \in \mathbb{F}_{64}^\times$. It can be easily seen that $\beta^2=1+\theta\beta$, $\beta^3=\theta+(1+\theta^2)\beta$ and $\beta^9=1$, so the polynomial $f(x)=x^9-1 \in \mathbb{F}_8[x]$ has zeros exactly given by $\cZ(f)=\{\beta^i: 0 \leq i \leq 8\}$.

\noindent{\bf{Example 8: MDS $(9, 6, 2)_8$}}

Consider the cyclic code $\cC$ of length 9 over $\mathbb{F}_8$ with generator polynomial $h(x)=\prod_{i=-2}^2 (x-\beta^i)$. As a BCH code, this code is MDS with parameters $[9, 4, 6]_8$. We use {\bf Construction one} in this case, and we have $k=4$ and $\delta=1$.

To give a parity-check matrix of $\cC$, we start from the matrix

$$\ol{H}=\begin{bmatrix}
1 & 1 & 1 & 1 & 1 & 1 & 1 & 1 & 1\\
1 & \beta & \beta^2 & \beta^3 & \beta^4 & \beta^5 & \beta^6 & \beta^7 & \beta^8\\
1 & \beta^2 & \beta^4 & \beta^6 & \beta^8 & \beta^{10} & \beta^{12} & \beta^{14} & \beta^{16}
\end{bmatrix}$$
\noindent{\footnotesize
\arraycolsep=3pt
\medmuskip = 1mu
$$=\begin{bmatrix}
1 & 1 & 1 & 1 & 1 & 1 & 1 & 1 & 1\\
1 & \beta & 1+\theta\beta & \theta+(1+\theta^2)\beta & (1+\theta^2)+(1+\theta)\beta & (1+\theta)+(1+\theta)\beta & (1+\theta)+(1+\theta^2)\beta & (1+\theta^2)+\theta\beta & \theta+\beta\\
1 & 1+\theta\beta & (1+\theta^2)+(1+\theta)\beta & (1+\theta)+(1+\theta^2)\beta & \theta+\beta & \beta & \theta+(1+\theta^2)\beta & (1+\theta)+(1+\theta)\beta & (1+\theta^2)+\theta\beta
\end{bmatrix}$$
}
$$=\begin{bmatrix}
\ul{k}_{1;0}\\
\ul{k}_{1;1}+\ul{k}_{1;-1}\beta\\
\ul{k}_{1;2}+\ul{k}_{1;-2}\beta
\end{bmatrix}$$
where each $\ul{k}_{1; i} \in \mathbb{F}_8^9, -2 \le i \le 2$ since $\{1, \beta\}$ is a basis of $\mathbb{F}_{64}=\mathbb{F}_8(\beta)$ over $\mathbb{F}_8$. A parity-check matrix of $\cC$ is given by
$$H=\begin{bmatrix}
\ul{k}_{1;0}\\
\ul{k}_{1;1}\\
\ul{k}_{1;-1}\\
\ul{k}_{1;2}\\
\ul{k}_{1;-2}
\end{bmatrix}$$
$$=\begin{bmatrix}
1 & 1 & 1 & 1 & 1 & 1 & 1 & 1 & 1\\
1 & 0 & 1 & \theta & 1+\theta^2 & 1+\theta & 1+\theta & 1+\theta^2 & \theta\\
0 & 1 & \theta & 1+\theta^2 & 1+\theta & 1+\theta & 1+\theta^2 & \theta & 1\\
1 & 1 & 1+\theta^2 & 1+\theta & \theta & 0 & \theta & 1+\theta & 1+\theta^2\\
0 & \theta & 1+\theta & 1+\theta^2 & 1 & 1 & 1+\theta^2 & 1+\theta & \theta
\end{bmatrix}.$$

Let $H_{0}$ be the first 3 rows of $H$ and $H_{1}$ be the rest 2 rows. Consider the convolutional code $V$ with parity-check matrix

$$H(D)=H_{0}+\widetilde{H}_{1}D=\begin{bmatrix}\ul{k}_{1;0}\\
\ul{k}_{1;1}+\ul{k}_{1;2}D\\
\ul{k}_{1;-1}+\ul{k}_{1;-2}D
\end{bmatrix}$$
\noindent{\footnotesize
\arraycolsep=3pt
\medmuskip = 1mu
$$=\begin{bmatrix}
1 & 1 & 1 & 1 & 1 & 1 & 1 & 1 & 1\\
1+D & D & 1+(1+\theta^2)D & \theta+(1+\theta)D & (1+\theta^2)+\theta D & 1+\theta & (1+\theta)+\theta D & (1+\theta^2)+(1+\theta)D & \theta+(1+\theta^2)D\\
0 & 1+\theta D & \theta+(1+\theta)D & (1+\theta^2)+(1+\theta^2)D & (1+\theta)+D & (1+\theta)+D & (1+\theta^2)+(1+\theta^2)D & \theta+(1+\theta)D & 1+\theta D
\end{bmatrix}$$
}

By Theorem \ref{5:thm3}, $V$ is strongly-MDS with parameters $(9, 6, 2)_8$ and also has a maximal distance profile. The free distance of $V$ is $6$.


\noindent{\bf{Example 9: MDS $(9, 5, 3)_8$}}

Let $\cC$ be the cyclic code of length 9 over $\mathbb{F}_8$ with generator polynomial $h(x)=\prod_{i=-3}^3 (x-\beta^i)$. This is a BCH code, and it is MDS with parameters $[9, 2, 8]_8$. We use {\bf Construction two} in this case, and we have $\tau=3$.

To give a parity-check matrix of $\cC$, we start from the matrix

$$\ol{H}=\begin{bmatrix}
1 & 1 & 1 & 1 & 1 & 1 & 1 & 1 & 1\\
1 & \beta & \beta^2 & \beta^3 & \beta^4 & \beta^5 & \beta^6 & \beta^7 & \beta^8\\
1 & \beta^2 & \beta^4 & \beta^6 & \beta^8 & \beta^{10} & \beta^{12} & \beta^{14} & \beta^{16}\\
1 & \beta^3 & \beta^6 & \beta^9 & \beta^{12} & \beta^{15} & \beta^{18} & \beta^{21} & \beta^{24}\\
\end{bmatrix}$$
\noindent{\footnotesize
\arraycolsep=3pt
\medmuskip = 1mu
$$=\begin{bmatrix}
1 & 1 & 1 & 1 & 1 & 1 & 1 & 1 & 1\\
1 & \beta & 1+\theta\beta & \theta+(1+\theta^2)\beta & (1+\theta^2)+(1+\theta)\beta & (1+\theta)+(1+\theta)\beta & (1+\theta)+(1+\theta^2)\beta & (1+\theta^2)+\theta\beta & \theta+\beta\\
1 & 1+\theta\beta & (1+\theta^2)+(1+\theta)\beta & (1+\theta)+(1+\theta^2)\beta & \theta+\beta & \beta & \theta+(1+\theta^2)\beta & (1+\theta)+(1+\theta)\beta & (1+\theta^2)+\theta\beta\\
1 & \theta+(1+\theta^2)\beta & (1+\theta)+(1+\theta^2)\beta & 1 & \theta+(1+\theta^2)\beta & (1+\theta)+(1+\theta^2)\beta & 1 & \theta+(1+\theta^2)\beta & (1+\theta)+(1+\theta^2)\beta
\end{bmatrix}$$
}
$$=\begin{bmatrix}
\ul{k}_{1;0}\\
\ul{k}_{1;1}+\ul{k}_{1;-1}\beta\\
\ul{k}_{1;2}+\ul{k}_{1;-2}\beta\\
\ul{k}_{1;3}+\ul{k}_{1;-3}\beta
\end{bmatrix},$$
where each $\ul{k}_{1; i} \in \mathbb{F}_8^9, -3 \le i \le 3$. Here the first index is kept as 1 since essentially those vectors for $-2 \le i \le 2$ are the same as those in {\bf Example 8}.

As $\{1, \beta\}$ is a basis of $\mathbb{F}_{64}=\mathbb{F}_8(\beta)$ over $\mathbb{F}_8$, a parity-check matrix of $\cC$ is given by
$$H=\begin{bmatrix}
\ul{k}_{1;0}\\
\ul{k}_{1;1}\\
\ul{k}_{1;-1}\\
\ul{k}_{1;2}\\
\ul{k}_{1;-2}\\
\ul{k}_{1;3}\\
\ul{k}_{1;-3}
\end{bmatrix}$$
$$=\begin{bmatrix}
1 & 1 & 1 & 1 & 1 & 1 & 1 & 1 & 1\\
1 & 0 & 1 & \theta & 1+\theta^2 & 1+\theta & 1+\theta & 1+\theta^2 & \theta\\
0 & 1 & \theta & 1+\theta^2 & 1+\theta & 1+\theta & 1+\theta^2 & \theta & 1\\
1 & 1 & 1+\theta^2 & 1+\theta & \theta & 0 & \theta & 1+\theta & 1+\theta^2\\
0 & \theta & 1+\theta & 1+\theta^2 & 1 & 1 & 1+\theta^2 & 1+\theta & \theta\\
1 & \theta & 1+\theta & 1 & \theta & 1+\theta & 1 & \theta & 1+\theta\\
0 & 1+\theta^2 & 1+\theta^2 & 0 & 1+\theta^2 & 1+\theta^2 & 0 & 1+\theta^2 & 1+\theta^2
\end{bmatrix}$$

Let $H_{0}=\begin{bmatrix}
\ul{k}_{1;1}\\
\ul{k}_{1;-1}\\
\ul{k}_{1;3}\\
\ul{k}_{1;-3}
\end{bmatrix}$ and $H_{1}=\begin{bmatrix}
\ul{k}_{1;0}\\
\ul{k}_{1;2}\\
\ul{k}_{1;-2}
\end{bmatrix}$. Consider the convolutional code $V$ with parity-check matrix

$$H(D)=H_{0}+\widetilde{H}_{1}D=\begin{bmatrix}
\ul{k}_{1;1}\\
\ul{k}_{1;-1}+\ul{k}_{1;0}D\\
\ul{k}_{1;3}+\ul{k}_{1;2}D\\
\ul{k}_{1;-3}+\ul{k}_{1;-2}D
\end{bmatrix}$$
\noindent{\footnotesize
\arraycolsep=3pt
\medmuskip = 1mu
$$=\begin{bmatrix}
1 & 0 & 1 & \theta & 1+\theta^2 & 1+\theta & 1+\theta & 1+\theta^2 & \theta\\
D & 1+D & \theta+D & (1+\theta^2)+D & (1+\theta)+D & (1+\theta)+D & (1+\theta^2)+D & \theta+D & 1+D\\
1+D & \theta+D & (1+\theta)+(1+\theta^2)D & 1+(1+\theta)D & \theta+\theta D & 1+\theta & 1+\theta D & \theta+(1+\theta)D & (1+\theta)+(1+\theta^2)D\\
0 & (1+\theta^2)+\theta D & (1+\theta^2)+(1+\theta)D & (1+\theta^2)D & (1+\theta^2)+D & (1+\theta^2)+D & (1+\theta^2)D & (1+\theta^2)+(1+\theta)D & (1+\theta^2)+\theta D
\end{bmatrix}$$
}

By Theorem \ref{51:thm4}, $V$ is MDS with parameters $(9, 5, 3)_8$. The free distance of $V$ is $8$. Since the block code with parity-check matrix $H_{0}$ is MDS, the code $V$ also has a maximal distance profile.

\noindent{\bf{Example 10: MDS $(9, 3, 2)_8$}}

Since $k=1$ is not congruent to $q$ modulo 2, we use $f(x)=x^9-\al^9 \in \mathbb{F}_8[x]$, so its set of zeros is exactly $\cZ(f)=\{\al\beta^i: 0 \leq i \leq 8\}$. Hence the block codes we consider are constacyclic codes of length 9 over $\mathbb{F}_8$. Here $\al$ can be taken to be any primitive element of $\mathbb{F}_{64}$ such that $\al^7=\beta$. In particular, since $\theta+\ga$ is a primitive element of $\mathbb{F}_{64}$ and $(\theta+\ga)^7=(\theta^{-1}\beta)^7=\beta^7$, so that $[(\theta+\ga)^4]^7=\beta$, we may take $\al=(\theta+\ga)^4=\theta^2+(1+\theta^2)\beta \in \mathbb{F}_{64}^\times$.

Consider the constacyclic code $\cC \subset \mathbb{F}_8[x]/(f(x))$ of length 9 over $\mathbb{F}_8$ with generator polynomial $h(x)=\prod_{i=-3}^4 (x-\al\beta^i)$. By Lemma 2, this code is MDS with parameters $[9, 1, 9]_8$. We have $k=1$ and $\delta=1$.

To give a parity-check matrix of $\cC$, we start from the matrix

$$\ol{H}=\begin{bmatrix}
1 & \al\beta & (\al\beta)^2 & (\al\beta)^3 & (\al\beta)^4 & (\al\beta)^5 & (\al\beta)^6 & (\al\beta)^7 & (\al\beta)^8\\
1 & \al\beta^2 & (\al\beta^2)^2 & (\al\beta^2)^3 & (\al\beta^2)^4 & (\al\beta^2)^5 & (\al\beta^2)^6 & (\al\beta^2)^7 & (\al\beta^2)^8\\
1 & \al\beta^3 & (\al\beta^3)^2 & (\al\beta^3)^3 & (\al\beta^3)^4 & (\al\beta^3)^5 & (\al\beta^3)^6 & (\al\beta^3)^7 & (\al\beta^3)^8\\
1 & \al\beta^4 & (\al\beta^4)^2 & (\al\beta^4)^3 & (\al\beta^4)^4 & (\al\beta^4)^5 & (\al\beta^4)^6 & (\al\beta^4)^7 & (\al\beta^4)^8
\end{bmatrix}$$
\noindent{\footnotesize
\arraycolsep=2pt
\medmuskip = 1mu
$$=\begin{bmatrix}
1 & (1+\theta^2)+(1+\theta^2)\beta & (1+\theta^2)\beta & (1+\theta+\theta^2)+\theta\beta & (1+\theta+\theta^2)+(1+\theta^2)\beta & 1+\theta^2\beta & (1+\theta+\theta^2)+(1+\theta)\beta & \theta+\beta & \theta^2+(1+\theta^2)\beta\\
1 & (1+\theta^2)+\theta^2\beta & 1+(1+\theta+\theta^2)\beta & \theta^2 & \theta+(\theta+\theta^2)\beta & \theta^2+\beta & \theta+\theta^2 & (1+\theta)+(1+\theta^2)\beta & (\theta+\theta^2)+\theta^2\beta\\
1 & \theta^2+(\theta+\theta^2)\beta & \theta^2+\theta^2\beta & (1+\theta)+\theta\beta & (1+\theta+\theta^2)\beta & \theta^2+\theta\beta & 1+(1+\theta)\beta & (1+\theta^2)+(1+\theta)\beta & (1+\theta)+(\theta+\theta^2)\beta\\
1 & (\theta+\theta^2)+(1+\theta)\beta & (1+\theta+\theta^2)+\beta & (1+\theta+\theta^2)+\theta\beta & \theta+\theta\beta & 1+(\theta+\theta^2)\beta & (1+\theta+\theta^2)+(1+\theta)\beta & 1+\theta\beta & (1+\theta)\beta
\end{bmatrix}$$
}
$$=\begin{bmatrix}
\ul{k}_{3;1}+\ul{k}_{3;0}\beta\\
\ul{k}_{3;2}+\ul{k}_{3;-1}\beta\\
\ul{k}_{3;3}+\ul{k}_{3;-2}\beta\\
\ul{k}_{3;4}+\ul{k}_{3;-3}\beta
\end{bmatrix}$$
where each $\ul{k}_{3; i} \in \mathbb{F}_8^9, -3 \le i \le 4$. As $\{1, \beta\}$ is a basis of $\mathbb{F}_{64}=\mathbb{F}_8(\beta)$ over $\mathbb{F}_8$, a parity-check matrix of $\cC$ is given by
$$H=\begin{bmatrix}
\ul{k}_{3;1}\\
\ul{k}_{3;0}\\
\ul{k}_{3;2}\\
\ul{k}_{3;-1}\\
\ul{k}_{3;3}\\
\ul{k}_{3;-2}\\
\ul{k}_{3;4}\\
\ul{k}_{3;-3}
\end{bmatrix}$$
$$=\begin{bmatrix}
1 & 1+\theta^2 & 0 & 1+\theta+\theta^2 & 1+\theta+\theta^2 & 1 & 1+\theta+\theta^2 & \theta & \theta^2\\
0 & 1+\theta^2 & 1+\theta^2 & \theta & 1+\theta^2 & \theta^2 & 1+\theta & 1 & 1+\theta^2\\
1 & 1+\theta^2 & 1 & \theta^2 & \theta & \theta^2 & \theta+\theta^2 & 1+\theta & \theta+\theta^2\\
0 & \theta^2 & 1+\theta+\theta^2 & 0 & \theta+\theta^2 & 1 & 0 & 1+\theta^2 & \theta^2\\
1 & \theta^2 & \theta^2 & 1+\theta & 0 & \theta^2 & 1 & 1+\theta^2 & 1+\theta\\
0 & \theta+\theta^2 & \theta^2 & \theta & 1+\theta+\theta^2 & \theta & 1+\theta & 1+\theta & \theta+\theta^2\\
1 & \theta+\theta^2 & 1+\theta+\theta^2 & 1+\theta+\theta^2 & \theta & 1 & 1+\theta+\theta^2 & 1 & 0\\
0 & 1+\theta & 1 & \theta & \theta & \theta+\theta^2 & 1+\theta & \theta & 1+\theta
\end{bmatrix}$$

Now let $H_{0}$ be the first 6 rows of $H$ and $H_{1}$ be the rest 2 rows. Consider the convolutional code $V$ with parity-check matrix

$$H(D)=H_{0}+\widetilde{H}_{1}D=\begin{bmatrix}
\ul{k}_{3;1}\\
\ul{k}_{3;0}\\
\ul{k}_{3;2}\\
\ul{k}_{3;-1}\\
\ul{k}_{3;3}+\ul{k}_{3;4}D\\
\ul{k}_{3;-2}+\ul{k}_{3;-3}D
\end{bmatrix}$$
\noindent{\footnotesize
\arraycolsep=2pt
\medmuskip = 1mu
$$=\begin{bmatrix}
1 & 1+\theta^2 & 0 & 1+\theta+\theta^2 & 1+\theta+\theta^2 & 1 & 1+\theta+\theta^2 & \theta & \theta^2\\
0 & 1+\theta^2 & 1+\theta^2 & \theta & 1+\theta^2 & \theta^2 & 1+\theta & 1 & 1+\theta^2\\
1 & 1+\theta^2 & 1 & \theta^2 & \theta & \theta^2 & \theta+\theta^2 & 1+\theta & \theta+\theta^2\\
0 & \theta^2 & 1+\theta+\theta^2 & 0 & \theta+\theta^2 & 1 & 0 & 1+\theta^2 & \theta^2\\
1+D & \theta^2+(\theta+\theta^2)D & \theta^2+(1+\theta+\theta^2)D & (1+\theta)+(1+\theta+\theta^2)D & \theta D & \theta^2+D & 1+(1+\theta+\theta^2)D & (1+\theta^2)+D & 1+\theta\\
0 & (\theta+\theta^2)+(1+\theta)D & \theta^2+D & \theta+\theta D & (1+\theta+\theta^2)+\theta D & \theta+(\theta+\theta^2)D & (1+\theta)+(1+\theta)D & (1+\theta)+\theta D & (\theta+\theta^2)+(1+\theta)D
\end{bmatrix}$$
}
By Theorem \ref{5:thm4}, $V$ is strongly-MDS with parameters $(9, 3, 2)_8$ and also has a maximal distance profile. The free distance of $V$ is $9$.

\noindent{\bf{Example 11: MDS $(9, 4, 2)_8$}}

We use the cyclic code $\cC$ as in {\bf Example 9}, but we instead use {\bf Construction one} in this case. We have $k=2$ and $\delta=1$.

We let $H_0$ be the first 5 rows of $H$ and $H_1$ be the rest 2 rows. Consider the convolutional code $V$ with parity-check matrix

$$H(D)=H_0+\widetilde{H}_1D=\begin{bmatrix}
\ul{k}_{1;0}\\
\ul{k}_{1;1}\\
\ul{k}_{1;-1}\\
\ul{k}_{1;2}+\ul{k}_{1;3}D\\
\ul{k}_{1;-2}+\ul{k}_{1;-3}D
\end{bmatrix}$$
\noindent{\footnotesize
\arraycolsep=3pt
\medmuskip = 1mu
$$=\begin{bmatrix}
1 & 1 & 1 & 1 & 1 & 1 & 1 & 1 & 1\\
1 & 0 & 1 & \theta & 1+\theta^2 & 1+\theta & 1+\theta & 1+\theta^2 & \theta\\
0 & 1 & \theta & 1+\theta^2 & 1+\theta & 1+\theta & 1+\theta^2 & \theta & 1\\
1+D & 1+\theta D & (1+\theta^2)+(1+\theta)D & (1+\theta)+D & \theta+\theta D & (1+\theta)D & \theta+D & (1+\theta)+\theta D & (1+\theta^2)+(1+\theta)D\\
0 & \theta+(1+\theta^2)D & (1+\theta)+(1+\theta^2)D & 1+\theta^2 & 1+(1+\theta^2)D & 1+(1+\theta^2)D & 1+\theta^2 & (1+\theta)+(1+\theta^2)D & \theta+(1+\theta^2)D
\end{bmatrix}$$
}
By Theorem \ref{5:thm3}, $V$ is strongly-MDS with parameters $(9, 4, 2)_8$ and also has a maximal distance profile. The free distance of $V$ is $8$.

\section{Conclusion}\label{conclude}

Maximum-distance separable (MDS) convolutional codes are characterized by the property that the free distance attains the generalized Singleton bound. Thus MDS convolutional codes form an optimal family of convolutional codes, the study of which is of great importance. There are very few general algebraic constructions of MDS convolutional codes. In this paper, we construct a large family of unit-memory MDS convolutional code over $\F$ of length $n \le q+1$ with flexible parameters. Compared with previous works, the field size $q$ required to define these codes is much smaller. The construction also leads to many new strongly-MDS convolutional codes, an important subclass of MDS convolutional codes proposed and studied in \cite{GL2}. At the end of the paper we present many examples to illustrate the construction.

\subsection*{Acknowledgment}
The research of M. Xiong was supported by RGC grant number 606211 and 609513 from Hong Kong.


\end{document}